\documentclass[preprint,oneside,british,12pt]{elsarticle}

\usepackage{graphicx}
\usepackage{tabularx}
\usepackage{natbib}
\usepackage{algorithm}
\usepackage{algpseudocode}

\usepackage{amsthm}
\usepackage{amsmath}
\usepackage{subfigure}
\usepackage{color}
\usepackage{multirow}

\biboptions{authoryear}

\newcommand{\B}{\boldsymbol}
\newcommand{\E}{\mbox{E}}
\newcommand{\sd}{\mbox{sd}}

\newcommand{\Prob}{\mbox{P}}
\newcommand{\X}{\mathcal{X}}
\newcommand{\N}{\mathcal{N}}
\newtheorem{proposition}{Proposition}

\newlength{\LL}

\begin{document}

\title{One family, six distributions -- A flexible model for insurance claim severity}

\author[uio]{Erik B{\o}lviken}
\ead{erikb@math.uio.no}
\author[uio]{Ingrid Hob{\ae}k Haff\corref{cor1}}
\ead{ingrihaf@math.uio.no}

\cortext[cor1]{Corresponding author}
\address[uio]{Department of Mathematics, University of Oslo, Postboks 1053 Blindern, 0316 Oslo, Norway}


\begin{abstract}
We propose a new class of claim severity distributions with six 
parameters, that has the standard two-parameter distributions, the 
log-normal, the log-Gamma, the Weibull, the Gamma and the Pareto, 
as special cases. This distribution is much more flexible than its 
special cases, and therefore more able to to capture important 
characteristics of claim severity data. Further, we have 
investigated how increased parameter uncertainty due to a larger
number of parameters affects the estimate of the reserve. This is 
done in a large simulation study, where both the characteristics of 
the claim size distributions and the sample size are varied. We 
have also tried our model on a set of motor insurance claims from 
a Norwegian insurance company. The results from the study show that 
as long as the amount of data is reasonable, the five- and 
six-parameter versions of our model provide very good estimates of 
both the quantiles of the claim severity distribution and the 
reserves, for claim size distributions ranging from medium to very 
heavy tailed. However, when the sample size is small, our model 
appears to struggle with heavy-tailed data, but is still adequate 
for data with more moderate tails. 
\end{abstract}

\titlepage
\maketitle

\clearpage

\section{Introduction}\label{sec:introduction}
The Burr family of loss distributions goes back to \cite{Burr1942} and is also 
called extended or generalized Pareto and Beta prime. Its probability density
functions are all unimodal (or monotonically decreasing). Further, there are 
three parameters which capture distributions of widely different shapes that 
range from heavy-tailed Paretos to light-tailed Gammas. The latter are on the 
boundary of the parameter space and are only reached as two of the parameters 
approach infinity. An even more versatile parametric family of distributions 
is proposed in this paper by applying parametrised power transformations to 
Burr-distributed random variables. We have used the classical BoxCox approach 
in~\cite{box64} with the modification that there are two free parameters. The 
resulting PowerBurr family has five parameters. It does again lead to 
unimodal probability density functions, and includes ten of the most commonly 
applied and quoted loss models in actuarial science, for example the log-normal 
and the Weibull in addition to the parent Burr distribution  itself. This will 
be verified in the next section where an even wider class with a sixth 
parameter also  will be introduced.

Constructions such as the preceding one raise the intriguing question of 
whether loss modelling can be carried out by fitting one of these 
many-parameter families and simply use the distribution found without any more 
ado, thus avoiding an additional model selection step. In situations where an 
underlying heterogeneity of unknown or unquantifiable source that cannot be 
linked to an observed covariate and expressed into a regression relationship,
more complex modelling may be needed. In such cases the natural model is a 
mixture distribution as in \cite{lee2010}, \cite{bakar2015} and 
\cite{miljkovic2016}, or even a traditional kernel density estimate based on, 
say Gaussian mixtures, as in \cite{scott92}. Yet, the most common of all types
of variations is undoubtedly the unimodal one, and here PowerBurr fitting is 
an alternative to the traditional approach of trying different two-parameter 
families and choosing between them by Q-Q plotting, formal selection criteria 
like AIC or BIC or goodness-of-fit testing. 

Our aim is not to decide whether PowerBurr fitting is superior or inferior to 
such methods in terms of the error in the final solution. However, we will 
argue that its conceptual simplicity is attractive and offers a useful 
potential for automating the entire process from historical data to statements 
of risk in the computer. With a standard Poisson model for claim frequency, 
$99\%$ or $99.5\%$ solvency capital may be be evaluated by Monte Carlo, after 
estimating claim frequency and fitting loss data. The entire process may then 
be carried out with no or almost no human intervention at all. Such a program 
is not hard to  implement in the computer and would  draw on the simple 
algorithm for simulating  PowerBurr variables in Section 2, but it  does 
require a robust numerical procedure for estimating the parameters, which is a 
challenge. The criterion might be quite flat in some of the parameters, perhaps 
with multiple optima which are important practical obstacles to be addressed 
below.

A flat likelihood function means that many alternative distributions are almost 
equally likely given historical losses. Loss distributions are in reality no 
more than tools for evaluating risk. Whether they reduce to simple families or 
not and the interpretation of their parameters may not necessarily be issues of 
primary importance. However, what {\em is} important is that the tail behaviour 
may be quite different for distributions that are almost equally well-fitting 
in the central domain containing most of the data. A traditional attempt to 
deal with this is extreme value mixing, drawing on the result due to 
\cite{pickands1975} that over-threshold distributions always become Pareto or 
exponentially distributed as the threshold becomes infinite; consult 
\cite{embrechts1997} for a review of such methods. A more recent contribution 
to such modelling is \cite{lee2012}. PowerBurr fitting may be an alternative 
even here, perhaps by replacing the likelihood by a criterion that emphasizes 
more strongly the fit in the extreme right tail, for example using a weighted
likelihood approach. This will however not be studied here. What will be 
investigated, is how errors in statements of risk, like the solvency capital, 
are linked to estimation errors in the PowerBurr family parameters. In 
particular, estimation uncertainty will increase with the number of parameters, 
which in the PowerBurr is quite large, as the amount of available data 
decreases. An extensive numerical experiment will be conducted, where the error
in $99\%$ solvency capital will be examined. $<$Litt mer her$>$.

\section{The PowerBurr family}\label{sec:sixpardist}

\subsection{Definition}\label{subsec:construction}
The most convenient route to the Burr family is to start with Gamma variables 
$G_\alpha$ and $G_\theta$ with expectation $1$ and shape parameters $\alpha>0$ 
and $\theta>0$. Their standard deviations are sd$(G_\alpha)=1/\sqrt{\alpha}$
and  sd$(G_\theta)=1/\sqrt{\theta}$ from which it follows that
$G_\alpha\overset{p}{\rightarrow} 1$ and $G_\theta\overset{p}{\rightarrow} 1$ as 
$\alpha\rightarrow \infty$ and $\theta\rightarrow\infty$.
The ratio
\begin{equation}
X = \frac{G_{\theta}}{G_{\alpha}}
\label{eqn:betaprime.constr}
\end{equation}
will be called a standard Burr$(\alpha,\theta)$ variable. Two useful 
observations are immediate. If $X$ is $\mbox{Burr}(\alpha,\theta)$, then $1/X$ 
is $\mbox{Burr}(\theta,\alpha)$; i.e the same model except for the role of the 
parameters being switched. Also note that $X$ becomes Gamma-distributed as
$\alpha\rightarrow \infty$ since $G_\alpha\overset{p}{\rightarrow} 1$. The 
construction below extends in many other directions. The probability density 
function of $X$ is given by
\begin{equation}
g(x) = 
\frac{\Gamma(\alpha + \theta)}{\Gamma(\alpha)\Gamma(\theta)}
\left(\frac{\alpha}{\theta}\right)^{\alpha}\frac{x^{\theta-1}}
{(\alpha/\theta+x)^{\theta+\alpha}}, \ \ x > 0 
\label{eqn:betaprime.pdf}
\end{equation}
where $\Gamma(\cdot)$ is the Gamma function. That was how the Burr family was 
defined originally in \cite{Burr1942}. The proof  of \eqref{eqn:betaprime.pdf}
is elementary; consult Section 9.7 in \cite{bolviken2014}. 

\begin{figure}[t]
\centering
\includegraphics[width=0.49\linewidth]{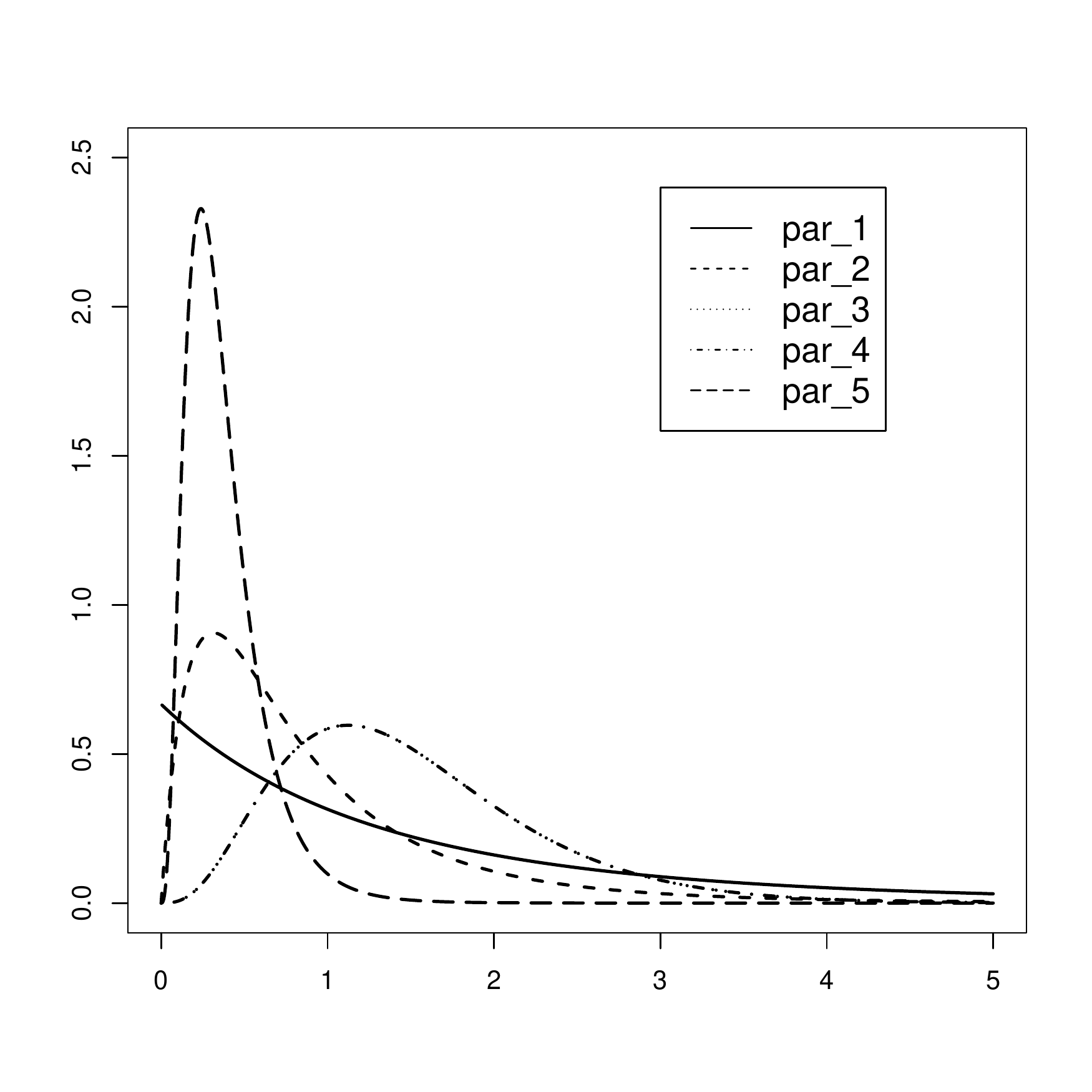}
\includegraphics[width=0.49\linewidth]{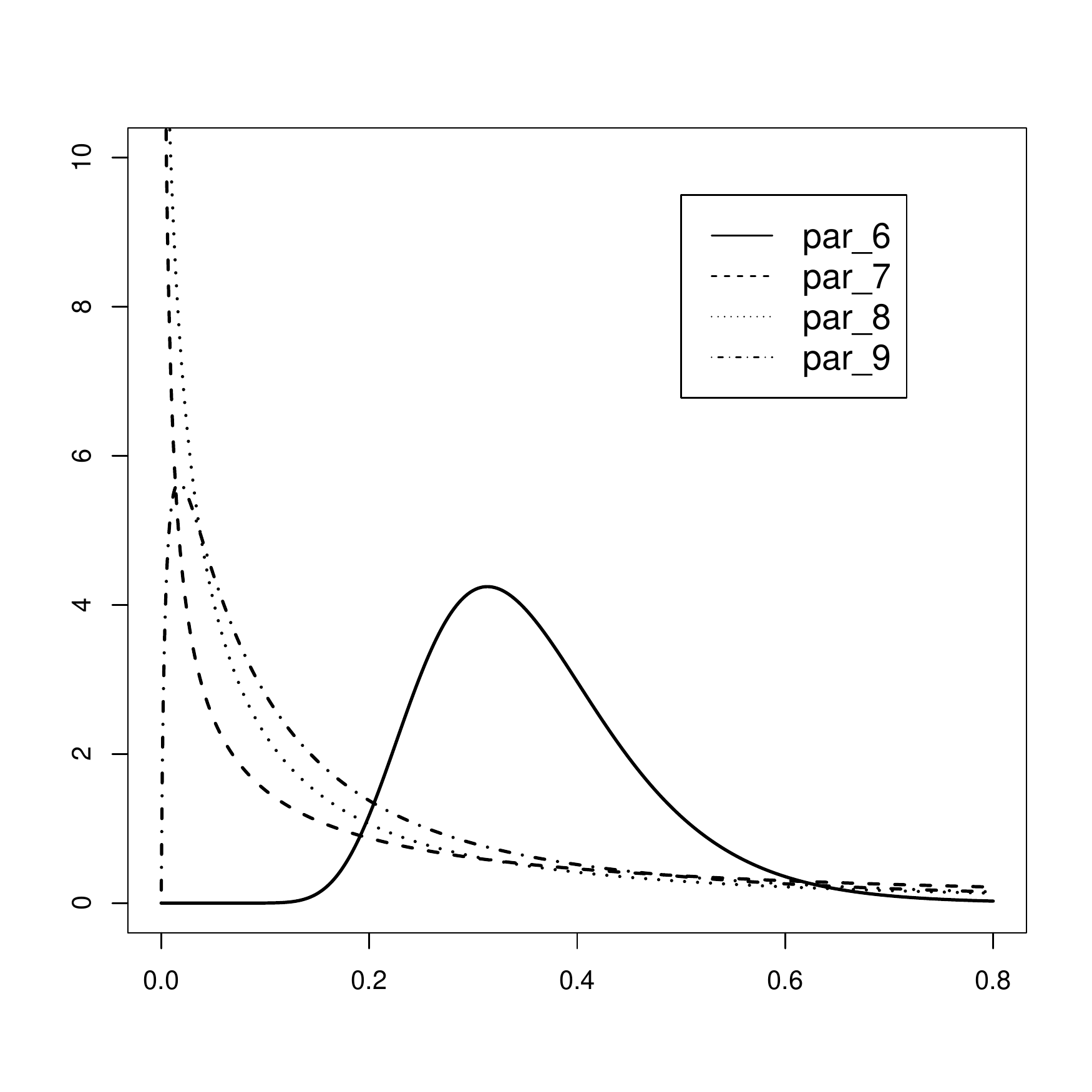}
\caption {Examples of PowerBurr probability density functions. 
The parameter vectors are
$\Phi_1=(5,1,1.5,1,1,1)$, $\Phi_2=(5,2,0.75,1,1,1)$, 
$\Phi_3=(1000,4,1.5,1,1,1)$, $\Phi_4=(1000,4,1000,1000,1.5,1)$, 
$\Phi_5=(1000,4,1,3250,1000,1)$, 
$\Phi_6=(1e6,1000,\exp(-10.6),1,10,100)$, 
$\Phi_7=(1000,1,0.4,2,1,1)$, 
$\Phi_8=(5,2,1.5,2,16,1.2)$ and $\Phi_9=(5,4,1.8,2,16,1.2)$.
}
\label{fig:sixparpdf}
\end{figure}
For the Burr model to be workable in practice it must be augmented with a 
parameter of scale $\beta>0$ so that it becomes
$Z=\beta X$. Our proposal is to replace this relationship with the more general 
\begin{equation}
Z = \beta\{(1 + X/\tau)^{\gamma}-1\}
\hspace*{1cm}\mbox{or even}\hspace*{1cm}
Z = \beta\{(1 + X^{\eta}/\tau)^{\gamma}-1\}
\label{eqn:sixpar.constr}
\end{equation}
with $\gamma>0$, $\tau>0$ and $\eta>0$ as additional parameters. The version on 
the left with parameter vector $\Phi=(\alpha,\theta,\beta,\tau,\gamma)$ will 
be referred to as PowerBurr$_5$, and represents the main thrust of this paper. 
The other, PowerBurr$_6$, has $\eta$ as a sixth parameter, and now
$\Phi=(\alpha,\theta,\beta,\tau,\gamma,\eta)$. In either case $Z>0$, and pure 
Burr reappears when  $\gamma=\tau=\eta=1$. The construction is inspired by the 
classical power transformation in \cite{box64}, but there $\tau=\gamma$ and 
$\eta=1$.

\subsection{Properties}\label{subsec:properties}
\hspace*{-0.6cm}{\em Unimodality} The power transformation in 
\eqref{eqn:sixpar.constr} is strictly increasing everywhere. This is however 
not a guarantee that when one uses it to transform a random variable, the 
probability density function of the new variable has a shape suitable for 
modelling. Yet, that is largely the case here, as the following proposition 
shows.
\begin{proposition} \label{prop:unimod}
 (i) All PowerBurr$_5$
and  PowerBurr$_6$ distributions are
 unimodal if $\gamma\geq 1$.
(ii) The PowerBurr$_5$ distributions are
 unimodal if $\max(\theta,\gamma)\geq 1$.
\end{proposition}
\noindent The calculations when the proposition is proved in Section \ref{subsec:proofs} 
shows that the precise condition for PowerBurr$_5$ distributions being unimodal 
is
\begin{equation}
\theta\geq 1\hspace*{1cm}\mbox{or}\hspace*{1cm}
\alpha(1-\gamma/\theta)-\tau(1+\alpha)\leq
2\sqrt{|1-\theta|\tau(\alpha+\gamma)\alpha/\theta}
\label{eqn:conunim}
\end{equation}
which may be too complicated to be of much practical value. If this fails to 
hold, the density function has both a local minimum and a maximum. Examples of 
probability density functions and their shapes are offered in Figure 
\ref{fig:sixparpdf}.
\\\\
{\em Moments and percentiles} 
Moments do not always exist. It emerges from expression 
\eqref{eqn:betaprime.pdf} that $E(X^r)$ is finite if $r<\alpha$, and since 
$X^{\eta\gamma}$ is the leading term in \eqref{eqn:sixpar.constr}, the condition 
for $E(Z^r)$ being finite is $r\eta\gamma<\alpha$. Closed mathematical 
expressions for these moments are not available unless $\theta=1$, and 
numerical methods are needed. Standard one-dimensional quadrature using 
\eqref{eqn:betaprime.pdf} is straightforward, and so is Monte Carlo (see 
below). Similarly, the percentiles of $X$ lack closed formulae unless 
$\theta=1$, and again numerical methods have to be used.
\\\\
{\em Sampling} It is easy to draw from  a PowerBurr distribution when a Gamma 
sampler is available. One simply generates Monte Carlo variates of $G_\theta$ 
and $G_\alpha$, compute their ratio $X$ and inserts that into 
\eqref{eqn:sixpar.constr}.
\\\\
{\em Special cases}
Most standard loss distributions in property insurance are included as special 
cases, which is summarized in the following proposition, verified in Section 
\ref{subsec:proofs}. 
\begin{proposition} \label{prop:speccases}
The PowerBurr$_5$ family contains the Burr, Pareto, Gamma, Inverse Gamma, 
Log-gamma, Logistic, Log-logistic, Weibull, Fr\'{e}chet, and the Log-normal 
families when the parameters $\Phi=(\alpha,\theta, \beta,\tau,\gamma)$ take the
values shown in Table \ref{tab:speccases}.
\end{proposition}
Nine of the ten families in Table \ref{tab:speccases} are defined in terms of 
shape parameters $a$ and $c$ and a parameter of scale $b$, some of them through stochastic representation, others through the survival function 
$\bar{F}(z)=\mbox{Pr}(Z>z)$. Only the definition of the Log-normal is 
different. The table shows how $\alpha$, $\theta$, $\beta$, $\tau$ and $\gamma$ 
must be defined for the various special cases to appear. Many of them are 
located at the boundary of the parameter space and are only reached in the 
limit. As an example consider the Log-gamma family on row $5$ which appears 
when $\tau=\gamma/b$, $\gamma\rightarrow\infty$ and $\alpha\rightarrow\infty$. 
This means that if PowerBurr$_5$ is fitted Log-gamma losses with enough data, 
the computer will return large values for all three of $\tau$, $\gamma$ and 
$\alpha$ with the scale parameter $b$ approximately the ratio $\gamma/\tau$.
\begin{table}[t]
\small
\begin{tabular}{rlcccccc}
\multicolumn{8}{l}{\footnotesize $^*$ For the log-normal
limit to be true $\alpha$ must tend to $\infty$ 
faster than $\theta$ in the sense that
$\theta/\alpha\rightarrow 0$.}\\
\hline
&&Definition&$\alpha$&$\theta$&$\tau$&$\gamma$&$\beta$\\
1&Burr&$Z=b G_c/G_a$&$a$&$c$&$1$&$1$&$b$\\
2&Pareto&$\bar{F}(z)=1/(1+z/b)^a$&$a$&$1$&$1$&$1$&$b$\\
3&Gamma&$Z=b G_c$&$\rightarrow \infty$&$c$&$1$&$1$&$b$\\
4&Inverse Gamma&$Z=b/G_a$&$a$&$\rightarrow \infty$
&$1$&$1$&$b$\\
5&Log-gamma&$\log(1+Z)=b G_c$&$\rightarrow \infty$&$c$&
$\gamma/b$&$\rightarrow \infty$&$1$\\
6&Logistic&$\bar{F}(z)=2/(1-a+ae^{z/b})$&$1$&$\rightarrow \infty$&$a$
&$\rightarrow 0$&$b/\gamma$\\
7&Log-logistic&$\bar{F}(z)=1/\{1+(z/b)^a\}$&$\rightarrow \infty$&$1$&
$\rightarrow 0$&$1/a$&$b\tau^\gamma$\\
8&Weibull&$Z=bG_1^{a}$&$1$&$\rightarrow \infty$&$\rightarrow 0$&
$a$&$b\tau^{\gamma}$\\
9&Fr\'{e}chet&$\bar{F}(z)=1-e^{-(z/b)^{-a}}$&$\rightarrow \infty$&1&
$\rightarrow 0$&$a$&$b\tau^{\gamma}$\\
10&Log-normal$^*$&$\log(Z)\sim N(\xi,\sigma)$&$\rightarrow\infty$&
$\rightarrow \infty$&$\sqrt{\theta}\sigma$&$\theta\sigma^2$&
$e^{-\sqrt{\theta}\sigma+\xi+1/2}$\\
\hline
\end{tabular}
\caption{Special cases  of PowerBurr$_5$ 
showing how its parameters are defined in terms of 
shape ($a$ and $c$) and scale ($b$)
of the parent distribution with the log-normal defined differently.}
\label{tab:speccases}
\end{table}
\\\\
{\em Identifiability} A stochastic model is identifiable if different parameter 
vectors $\Phi$ always lead to different distributions. That is not the case 
here. For example, if  $\gamma=1$,~(\ref{eqn:sixpar.constr}) shows that $\beta$ 
and  $\tau$ affect $Z$ through their ratio $\beta/\tau$. It follows that if 
$\gamma=1$, the likelihood method in the next section will pick a pair 
$(\beta,\tau)$ with the right ratio. Which  one does not matter since we in 
applications rarely are interested in the parameters themselves.

\subsection{Fitting}\label{subsec:fitting}
Models have in this paper been fitted by maximum likelihood, which requires
a mathematical expression for the probability density function $f(z;\Phi)$ of 
$Z$. Using the logarithmic form in \eqref{eqn:P0} below and inserting 
\eqref{eqn:betaprime.pdf} and \eqref{eqn:sixpar.constr}, straightforward
calculations lead to
\begin{align}
\begin{split}
\log\{f(z;\Phi)\} = &\log\{C(\Phi)\}+(\theta-\eta)\log(x)-(\alpha+\theta)\log(\alpha/\theta+x)\\
&+(\gamma-1)\log(1+x^\eta/\tau)
\end{split}
\label{eqn:loglike1}
\end{align}
where
\begin{equation}
C(\Phi)=\frac{\Gamma(\alpha+\theta)\tau(\alpha/\theta)^{\alpha}}
{\Gamma(\alpha)\Gamma(\theta)\gamma\beta\eta}
\hspace*{1cm}\mbox{and}\hspace*{1cm}x=\{(z/\beta+1)^{1/\gamma}-1\}\tau^{1/\eta}.
\label{eqn:loglike2}
\end{equation}
The relationship between $z$ and $x$ is the inverse of 
\eqref{eqn:sixpar.constr}. When historical losses $z_1,\dots,z_n$ are given, 
$\log\{f(z_i;\Phi)\}$ must be computed for each observation and added for the 
log-likelihood function
\begin{equation}
{\cal L}(\Phi;z_1,\dots,z_n)=\log\{f(z_1;\Phi)\}+\dots+\log\{f(z_n;\Phi)\}
\label{loglike3}
\end{equation}

Optimization was in this paper carried out by the quasi-Newton option in the 
R-function \texttt{optim()}. This is straightforward for two-parameter models,
but more challenging when there are five or even six parameters and a 
log-likelihood function that may be quite flat, especially when $n$ is 
moderate. This means that widely different parameter values result in rather 
similar distributions. Consequently, the optimisation of the likelihood is 
sometimes challenging, but this is not that problematic when the resulting 
distributions are sensible. Still, we have derived the derivatives of the 
log-likelihood functions with respect to all parameters, and supply those to 
\texttt{optim()}. These are given in the Appendix. To ease the optimisation 
further, it is performed on the log-transformed parameters as all the 
parameters $\alpha,\theta,\beta,\eta,\tau,\gamma$ are positive. Finally, we do 
the optimisation with several sets of start values for the parameters, and 
choose the parameter estimates that give the highest likelihood value. More 
numerical details are given in Appendix 2.

\subsection{Proofs of propositions}\label{subsec:proofs}
\begin{proof}[Proof of Proposition \ref{prop:unimod}] 
It is convenient to drop the parameter vector  and let $f(z)$ be the 
probability density function for $Z$ rather than $f(z;\Phi)$ as above. Its 
relationship to the probability density $g(x)$ for $X$ is through
\begin{align*}
f(z)=g(x)\,\frac{dx}{dz}
\end{align*}
and since $dx/dz=(dz/dx)^{-1}$, it follows that
\begin{equation}
\log\{f(z)\}=\log\{g(x)\}-\log\left(\frac{dz}{dx}\right)
\label{eqn:P0}
\end{equation}
so that
\begin{equation}
\frac{d\log\{f(z)\}}{dz}=
\frac{dx}{dz}\left(\frac{d\log\{g(x)\}}{dx}-\frac{d\log(dz/dx)}{dx}\right).
\label{eqn:P1}
\end{equation}
Recall that $g(x)=cx^{\theta-1}/(\alpha/\theta+x)^{\theta+\alpha}$ where $c$ is a 
constant. Hence
\begin{equation}
\frac{d\log\{g(x)\}}{dx}=\frac{\theta-1}{x}-\frac{\theta+\alpha}
{\alpha/\theta+x},
\label{eqn:P2}
\end{equation}
and moreover $z=\beta\{(1+x^\eta/\tau)^\gamma-1\}$ so that
\begin{align*}
\frac{dz}{dx}=\frac{\beta\gamma\eta}{\tau}(1+x^\eta/\tau)^{\gamma-1}x^{\eta-1}
\end{align*}
and
\begin{align*}
\log\left(\frac{dz}{dx}\right)=\log\left(\frac{\beta\gamma\eta}{\tau}\right)
+(\gamma-1)\log(1+x^\eta/\tau)+(\eta-1)\log(x).
\end{align*}
Differentiating this yields
\begin{equation}
\frac{d\log(dz/dx)}{dx}=\frac{(\gamma-1)\eta x^{\eta-1}/\tau}{1+x^\eta/\tau}
+\frac{\eta-1}{x}
\label{eqn:P3}
\end{equation}
which is with \eqref{eqn:P2}) to be inserted into \eqref{eqn:P1}. Then
\begin{align*}
\frac{d\log\{f(z)\}}{dz}=\frac{dx}{dz}\left(
\frac{\theta-\eta}{x}-\frac{\theta+\alpha}{\alpha/\theta+x}
+\frac{(1-\gamma)\eta x^{\eta-1}/\tau}{\tau+x^{\eta}}\right),
\end{align*}
and after some straightforward manipulations this can be rewritten
\begin{align*}
\frac{d\log\{f(z)\}}{dz}=\frac{dx}{dz}\,\frac{P(x)}{x(\alpha/\theta+x)(\tau+x^\eta)}
\end{align*}
where
\begin{align*}
P(x) = -(\alpha+\gamma\eta)x^{\eta+1}+\alpha(1-\gamma\eta/\theta)x^{\eta}-\tau(\eta+\alpha)x+\tau(\theta-\eta)\alpha/\theta.
\end{align*}
The sign of the derivative is determined by $P(x)$ since the all other factors 
are positive. A more useful form is
\begin{align*}
P(x) = \alpha(x^{\eta}+\tau)(1-x)-\eta(\gamma x^{\eta}+\tau)(x+\alpha/\theta)
\end{align*}
which is established by multiplying out both factors and collecting the terms.
A solution of the equation $P(x) = 0$ must therefore satisfy
\begin{equation}
\eta\frac{\gamma x^{\eta}+\tau}{x^{\eta}+\tau} = \alpha\frac{1-x}{x+\alpha/\theta}.
\label{eqn:P4}
\end{equation}
When $\gamma \geq 1$, the left hand side of \eqref{eqn:P4} is monotonically 
increasing in $x$ whereas the right hand side is monotonically decreasing, and 
there is at most one solution to the equation. This means that 
$d\log\{f(z)\}/dz$ has zero or one root, and $f(z)$ is monotonically decreasing 
everywhere or has a single mode. The latter applies when the maximum of the 
right hand side of \eqref{eqn:P4} is larger than the minimum on the left, i.e. 
when $\theta > \eta$. This verifies  the first part of Proposition 
\ref{prop:unimod}. 

To address the second part concerning PowerBurr$_5$ let  $\eta=1$. Now $P(x)$ 
becomes
\begin{align*}
P(x) = -(\alpha+\gamma)x^2+\{\alpha(1-\gamma/\theta)-
\tau(1+\alpha)\}x+\tau(\theta-1)\alpha/\theta.
\end{align*}
which is a second order polynomial with negative quadratic term. Such functions 
have a single maximum, and if $\theta\geq 1$ so that $P(0)\geq 0$, there can 
be no more than one solution of the equation $P(x)=0$  so that unimodality 
still holds. The condition can be expressed as $\max(\theta,\gamma)\geq 1$, as 
claimed in part (ii) of Proposition \ref{prop:unimod}.

To derive the precise algebraic condition \eqref{eqn:conunim} we must examine 
$P(x)$ in detail. Elementary calculations show that its maximum occurs at
\begin{align*}
x_m = \frac{\alpha(1-\gamma/\theta)-\tau(1+\alpha)}
{2(\alpha+\gamma)}
\end{align*}
with maximizing value
\begin{align*}
P(x_m) = \frac{\{\alpha(1-\gamma/\theta)-\tau(1+\alpha)\}^2}
{4(\alpha+\gamma)}
+\tau(\theta-1)\alpha/\theta.
\end{align*}
The  condition for the equation $P(x)=0$ to have a single, positive  root or no 
positive root at all is that either $P(0)\geq 0$ so that $\theta\geq 1$ as in 
\eqref{eqn:conunim} left or either of $x_m\leq 0$ and $P(x_m)\leq 0$ being 
true. The latter pair of conditions can be expressed jointly as
\begin{align*}
\alpha(1-\gamma/\theta)-\tau(1+\alpha)\leq
2\sqrt{|1-\theta|\tau(\alpha+\gamma)\alpha/\theta}
\end{align*}
as in \eqref{eqn:conunim} right.
\end{proof}
\begin{proof}[Proof of Proposition \ref{prop:speccases}]
One immediately sees that the Burr and Pareto distributions themselves are 
included in the PowerBurr family. The Gamma distribution is obtained by letting 
$\alpha\rightarrow \infty$ so that $G_\alpha\overset{p}{\rightarrow} 1$. By 
Slutsky's theorem this implies $X\overset{d}{\rightarrow}\beta G_\theta$ which 
defines the Gamma family. The Inverse Gamma is obtained in a similar way. Now 
$G_\theta\overset{p}{\rightarrow} 1$ as  $\theta\rightarrow \infty$ and 
$X\overset{d}{\rightarrow}\beta/G_\alpha$ on row $4$ in Table 
\ref{tab:speccases}. 

In order to get the Log-Gamma or row $5$ insert $\beta=1$ and $\tau=\gamma/b$ 
into \eqref{eqn:sixpar.constr}. Then
\begin{displaymath}
Z=(1+bX/\gamma)^\gamma-1\rightarrow e^{bX}-1
\hspace*{1cm}\mbox{as}\hspace*{1cm} \gamma\rightarrow \infty,
\end{displaymath}
and it follows that $\log(1+Z)\overset{p}{\rightarrow}bX$. Since
$X\overset{p}{\rightarrow}G_\theta$ as $\alpha\rightarrow \infty$,
$\log(1+Z)\overset{p}{\rightarrow}bG_\theta$ and
$\log(1+Z)\overset{d}{\rightarrow}bG_\theta$. 

The logistic distribution on row $6$ emerges when $\beta=b/\gamma$ so that by 
an elementary application of l'H\^{o}pital's rule
\begin{displaymath}
Z=b\,\frac{(1+X/\tau)^\gamma-1}{\gamma}\rightarrow b\log(1+X/\tau)
\hspace*{1cm}\mbox{as}\hspace*{1cm} \gamma\rightarrow 0,
\end{displaymath}
which yields
\begin{displaymath}
\mbox{Pr}(Z>z)\rightarrow\mbox{Pr}(\log(1+X/\tau)> z/b)
=\mbox{Pr}\left(X>\tau(e^{z/b}-1)\right).
\end{displaymath}
Let $\alpha=\theta=1$ so that $X$ follows an ordinary Pareto distribution with 
parameters $a=b=1$. As remarked in in Section \ref{subsec:construction}, this
is also the distribution of $1/X$ which means that the limit on the right is
$<$Er det en $1-$ for mye i første likhet?$>$
\begin{displaymath}
\mbox{Pr}(1/X>\tau(e^{z/b}-1))=1-\frac{1}{1+\tau^{-1}(e^{z/b}-1)^{-1}}
=\frac{1}{1-\tau+\tau(e^{z/b}-1)}
\end{displaymath}
and in conclusion
\begin{displaymath}
\mbox{Pr}(Z>z)\rightarrow \frac{1}{1-\tau+\tau(e^{z/b}-1)}
\hspace*{1cm}\mbox{as}\hspace*{1cm} \gamma\rightarrow 0.
\end{displaymath}
The logistic model on row $6$ follows when $\tau=1/2$. This is the definition 
used in \cite{beir96}, but one could also contemplate an extended version with 
$\tau$ as a general parameter.

To obtain the log-logistic, the Weibull and the Frech\`{e}t on rows $7$, $8$ 
and $9$ one starts with  $\beta=b\tau^{\gamma}$. Then 
\begin{displaymath}
Z=b\tau^\gamma\{(1+X/\tau)^\gamma-1\}\rightarrow bX^\gamma
\hspace*{1cm}\mbox{as}\hspace*{1cm} \tau\rightarrow 0.
\end{displaymath}
When $\theta=\alpha=1$, $X$ follows an ordinary Pareto distribution, such that
\begin{displaymath}
\mbox{Pr}(Z>z)\rightarrow \mbox{Pr}(X>(z/b)^{\gamma})=
\frac{1}{1+(z/b)^{1/\gamma}}
\end{displaymath}
which is the expression on row $7$ when $\gamma=1/a$. 

The Weibull model is defined as $Z=bG_1^a$, where $G_1$ is the exponential 
distribution with mean $1$, and emerges when $\theta=1$, 
$\alpha\rightarrow \infty$, $\gamma=a$ and $\tau\rightarrow 0$.

The Frech\`{e}t is obtained with the same specifications except that now 
$\theta\rightarrow \infty$ and $\alpha=1$. Then
\begin{displaymath}
\mbox{Pr}(Z>z)\rightarrow \mbox{Pr}(1/G_1>(z/b)^{\gamma})=
1-e^{-(z/b)^{-a}}
\end{displaymath}
as on row $9$.

Finally, to get the Log-normal on row $10$, let $\beta=e^{\xi+1/2-\sqrt{\gamma}}$ 
and $\tau=\sqrt{\gamma}$ so that
\begin{displaymath}
Z=\beta(1+X/\tau)^\gamma-\beta=e^{\xi+1/2-\sqrt{\gamma}+\gamma\log(1+X/\sqrt{\gamma})}-\beta.
\end{displaymath}
Taylor's formula with two terms and remainder yields
\begin{displaymath}
\log(1+X/\sqrt{\gamma})=X/\sqrt{\gamma}-X^2/(2\gamma)+R
\hspace*{1cm}\mbox{where}\hspace*{1cm}
0<R<\frac{1}{6}\frac{X^3}{\gamma^{3/2}}.
\end{displaymath}
The bounds for $R$ are consequences of $X>0$ and imply 
$\gamma R\overset{p}{\rightarrow} 0$ as $\gamma\rightarrow\infty$. It now 
follows after a little reorganization that
\begin{equation}
Z=e^{\xi+\sqrt{\gamma}(X-1)+(X^2-1)/2+\gamma R}-\beta.
\label{eqn:P5}
\end{equation}
Insert $\gamma=\theta\sigma^2$ so that
\begin{displaymath}
\sqrt{\gamma}(X-1)=\sigma\sqrt{\theta}\left(\frac{G_\theta}{G_\alpha}-1\right)
=\sigma\frac{\sqrt{\theta}(G_\theta-1)+(\sqrt{\theta/\alpha})
\sqrt{\alpha}(G_\alpha-1)}{G_\alpha}.
\end{displaymath}
As $\theta\rightarrow\infty$ and $\alpha\rightarrow\infty$, 
$\sqrt{\theta}(G_\theta-1)\overset{d}{\rightarrow} N(0,1)$ and 
$\sqrt{\alpha}(G_\alpha-1)\overset{d}{\rightarrow} N(0,1)$. In the denominator,
$G_\alpha\overset{p}{\rightarrow} 1$ from which it follows by Slutsky's theorem
that $\sqrt{\gamma}(X-1)\overset{d}{\rightarrow} N(0,\sigma)$ if both $\theta$ 
and $\alpha$ $\rightarrow \infty$ and $\theta/\alpha\rightarrow 0$. Further, 
$\beta\rightarrow 0$ as $\gamma\rightarrow\infty$. Since 
$X^2-1\overset{p}{\rightarrow} 0$ and $\gamma R\overset{p}{\rightarrow} 0$, the 
exponent in \ref{eqn:P5} tends to $N(\xi,\sigma)$ in distribution, as was to be 
proved.
\end{proof}
\section{Total loss model}\label{sec:model}

As mentioned in the introduction, we assume the classic collective model for
the total loss $\X$, given by \eqref{eqn:total_loss}, where the claim sizes
$Z_{i}$ are independent and identically distributed, and also independent of
the number $\N$ of claims. Assuming that $Z_{i}$s from different policies 
are identically distributed is obviously a simplification. However, it may
be argued that when the main interest is measuring risk in terms of the
total loss $\X$, and not pricing of individual contracts, differences
between policies will be averaged out.

Further, we assume that the number of claims follows a Poisson distribution
with fixed intensity $\mu$ for all policies, i.e. $\N \sim Poisson(\mu A)$,
where $A$ is the total risk exposure of the portfolio, i.e. the number of
policy years in the portfolio during the priod over which the claims are
counted. This is also a simplification, but as noted earlier, the main focus
of our paper is on the claim size distribution. 

Finally, the claims sizes are assumed to follow a distribution with
parameters $\B{\theta}$ and pdf $f_{Z}(z;\B{\theta})$, more specifically the
six-parameter distribution or one of its special cases. The model for the
total loss is then such that there is no explicit expression for the reserve
for any of the claim severity distributions we consider. Therefore, these 
have to be estimated with Monte Carlo methods. The first step is then to 
generate $m$ independent samples $\X_{1}^{*},\ldots,\X_{m}^{*}$ from the 
model, as specified in Algorithm \ref{alg:reserve_cte}. Subsequently, the 
reserve is estimated by
\begin{equation}
q_{\epsilon}^{*} = \X_{((1-\epsilon)m)}^{*},
\label{eqn:var}
\end{equation}
where $\X_{(1)}^{*}\leq\ldots\leq\X_{(m)}^{*}$. In this
setting, the total risk exposure is $A=JT$, where $J$ is the current number
of policies in the portfolio and $T=1$ year. In addition, we are interested
in the $100\cdot(1-\epsilon)\%$ quantile $q_{\epsilon}^{Z}$ of the claim
size distribution, given by
\[
\Prob(Z > q_{\epsilon}^{Z}) = \epsilon.
\]
This will indicate how well our model captures the tail properties of the
claim severity distribution and may help us to understand the performance
of our model for estimating the reserve. For the extended Pareto distribution,
we compute $q_{\epsilon}^{Z}$ with a numeric search algorithm and for the 
six-parameter distribution, it is computed by \eqref{eqn:quantfunc}.

Even though the six-parameter distribution is very flexible, it will never
be the exact true data generating process $g(\cdot)$. Hence, there will always be 
some model error in estimates of the reserve from this distribution. If
the parameters $\B{\theta}$ are estimated by the maximum likelihood,
the estimate $\hat{\B{\theta}}$ will converge to the value of $\B{\theta}_{0}$
that minimises the Kullback-Leibler divergence 
\[
KL(g,f) = \int_{0}^{\infty}log\left(\frac{g(z)}{f(z;\B{\theta})}\right)g(z)dz
\]
from the true $g$ to the model $f$. Now, let $q_{\epsilon}(\B{\theta})$ be the
reserve based on $f(\cdot;\B{\theta})$, and $q_{\epsilon}^{*}(\hat{\B{\theta}})$
the corresponding Monte Carlo estimate. Then we have:
\begin{align*}
q_{\epsilon}^{*}(\hat{\B{\theta}})-q_{\epsilon} = & (q_{\epsilon}^{*}(\hat{\B{\theta}})-q_{\epsilon}(\hat{\B{\theta}}))+(q_{\epsilon}(\hat{\B{\theta}})-q_{\epsilon}(\B{\theta}_{0}))+(q_{\epsilon}(\B{\theta}_{0})-q_{\epsilon})\\
= & E_{1}+E_{2}+E_{3},
\end{align*}
where $E_{1}$ is the Monte Carlo error, $E_{2}$ is the estimation error and 
$E_{3}$ is the model error. The former, $E_{1}$, depends on the number of 
simulations $m$, and will converge to $0$ as $m$ tends to $\infty$. Hence,
the Monte Carlo error will be quite small as long as one chooses a large
enough $m$. Further, the model error $E_{3}$ depends on the Kullback-Leibler 
divergence form the true distribution to the assumed distribution. As the
six-parameter distribution is very malleable, this error is also likely
to be small. We may thus assume that the estimation error $E_{2}$, which is 
a function of the sample size, is the main source of error.  

\begin{algorithm}[H]
	\caption{\label{alg:reserve_cte}}
	Input: $J$, $\mu$, $T$, $\B{\theta}$, $f_{Z}(\cdot;\B{\theta})$, $m$
	\begin{algorithmic}[1]
		\For{$i = 1,\ldots,m$}
		\State $\X_{i}^{*} = 0$
		\State Draw $\N^{*} \sim Poisson(J\mu T)$
		\For{$j = 1,\ldots,\N^{*}$}
		\State Draw $Z^{*}$ from $f_{Z}(z;\B{\theta})$
		\State $\X_{i}^{*} = \X_{i}^{*}+Z^{*}$
		\EndFor
		\EndFor
		\State Return $\X_{1}^{*},\ldots,\X_{m}^{*}$
	\end{algorithmic}
\end{algorithm}

\section{Simulation study}\label{sec:simstudy}

The proposed claim size distribution is more flexible than its two-parameter
special cases. The question is how well it estimates the reserve, and how this 
estimate is affected by the increased parameter uncertainty due to the larger 
number of parameters. In order to investigate this, we conduct a large simulation 
study. The parameter settings for the study are presented in Section 
\ref{subsec:parset}, some remarks about parameter estimation are given in 
Section \ref{subsec:estimation} and the results from the study are given in 
Section \ref{subsec:results}.

\subsection{Parameter settings}\label{subsec:parset}

In all, we consider $10$ different claim size distributions: the log-normal,
the log-Gamma, the Weibull, the Pareto, the Gamma, the extended Pareto, the 
four-parameter special case with $\beta=\tau=1$, the two five-parameter special 
cases with $\eta=1$ and $\tau=1$, respectively, and the full 
six-parameter distribution. Each experiment is performed as described in
Algorithm \ref{alg:exp}. One of the $10$ distributions is the true claim 
size distribution with pdf $f_{Z}(\cdot;\B{\theta})$. First, a sample 
$z_{1},\ldots,z_{n}$ of size $n$ is drawn from this distribution. Then, the 
parameters of each of the $10$ distributions considered in this study are 
estimated based on the sample, resulting in the estimates $\hat{\B{\theta}}_{1},\ldots,\hat{\B{\theta}}_{10}$. Finally, one estimates 
the reserve $q_{\epsilon}$ and the quantile $q_{\epsilon}^{Z}$ of the
claim size distribution using Algorithm \ref{alg:reserve_cte} with each 
$\hat{\B{\theta}}_{i}$ and corresponding $f_{Z,i}(\cdot;\hat{\B{\theta}}_{i})$. 
This procedure is repeated $N$ times. Note that the parameter of the claim 
frequency distribution is not estimated, but set to its true value. This is 
done in order to isolate the uncertainty in the reserve stemming from the 
claim size distribution.

\begin{algorithm}[H]
	\caption{\label{alg:exp}}
	Input: $J$, $\mu$, $\B{\theta}$, $f_{Z}(\cdot;\B{\theta})$, $m$
	\begin{algorithmic}[1]
		\For{$k = 1,\ldots,N$}
		\For{$j = 1,\ldots,n$}
		\State Draw $Z^{*}$ from $f_{Z}(z;\B{\theta})$
		\EndFor
		\For{$i = 1,\ldots,10$}
		\State Estimate $\hat{\B{\theta}}_{i}$
		\State Estimate $(\hat{q}_{\epsilon,k,i}^{*},\hat{q}_{\epsilon,k,i}^{Z}$ using Algorithm \ref{alg:reserve_cte} and \eqref{eqn:var} with $\B{\theta}=\hat{\B{\theta}}_{i}$.
		\EndFor
		\EndFor
		\State Return $(\hat{q}_{\epsilon,k,i}^{*},\hat{q}_{\epsilon,k,i}^{Z}), \ i=1,\ldots,10, \ k=1,\ldots,N$.
	\end{algorithmic}
\end{algorithm}

The parameter values used in the study are shown in Table 
\ref{tab:parameters}. These are chosen  so that 
$\E(Z) \approx 1$, except for the log-Gamma distribution, for which this 
was impossible without making $\sd(Z)$ very small. Further, the parameter
values are set so that the corresponding distributions range from medium
to very heavy-tailed. Figure \ref{fig:simstudypdfs} shows the corresponding 
pdfs. For the claim frequency distribution (which is Poisson), we let 
$\lambda=J\mu T$ take each of the values $10, \ 100, \ 1,000,$ to 
represent variations in portfolio size and claim intensity.

\begin{table}[t]
\begin{center}
\begin{tabular}{llcc}
\hline
\hline
\textsc{Distribution} & \textsc{$\B{\theta}$} & \textsc{$\E(Z)$} & \textsc{$\sd(Z)$}\\
\hline
$Log-normal(\xi,\sigma)$ & $(-0.5,1)$ & $1.00$ & $1.31$\\ 
$Log-Gamma(\xi,\theta)$ & $(0.75,5)$ & $1.25$ & $0.93$\\ 
$Weibull(\beta,\eta)$ & $(2,1.13)$ & $1.00$ & $0.52$\\ 
$Pareto(\alpha,\beta)$ & $(3,2)$ & $1.00$ & $1.77$\\ 
$Gamma(\xi,\alpha)$ & $(1,2)$ & $1.00$ & $0.71$\\ 
$Ext. Pareto(\alpha,\theta,\beta)$ & $(3,2,1)$ & $1.50$ & $2.22$\\ 
$Four-parameter(\alpha,\theta,\beta,\eta)$ & $(4,2,0.6,1.3)$ & $1.02$ & $1.81$\\ 
$Five-parameter(\alpha,\theta,\beta,\tau,\gamma)$ & $(4,2,2.7,5,1.3)$ & $1.00$ & $1.27$\\ 
$Five-parameter2(\alpha,\theta,\beta,\eta,\gamma)$ & $(4,2,0.5,1.2,1.1)$ & $0.94$ & $1.60$\\ 
$Four-parameter(\alpha,\theta,\beta,\eta,\tau,\gamma)$ & $(4,2,4,1.3,10,1.2)$ & $0.86$ & $1.85$\\
\hline
\hline
\end{tabular}
\caption{Parameters used for the $10$ distributions in the simulation study, as well as the mean and standard deviation of the corresponding distributions.}
\label{tab:parameters}
\end{center}
\end{table}

\begin{figure}[t]
\centering
\includegraphics[width=0.49\linewidth]{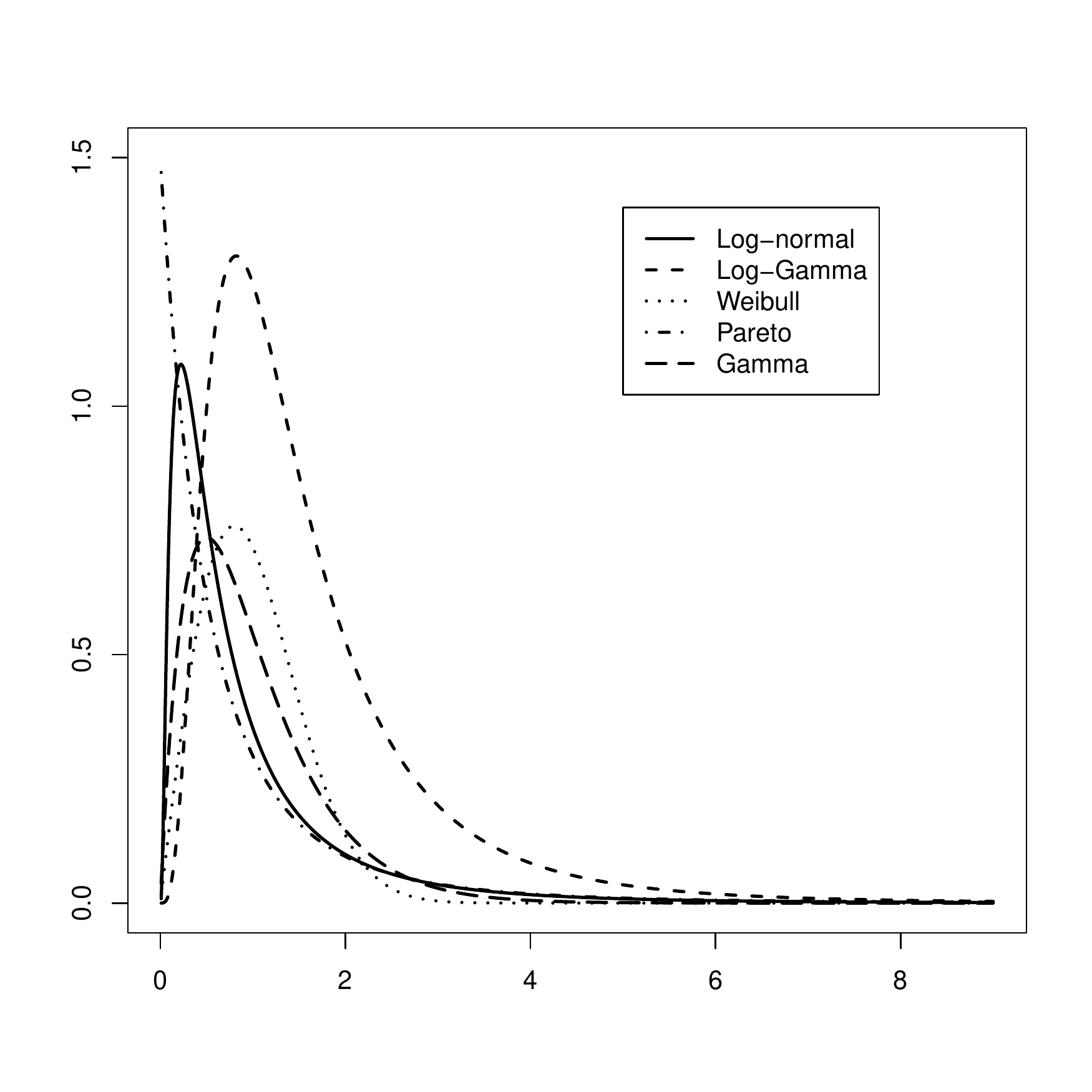}
\includegraphics[width=0.49\linewidth]{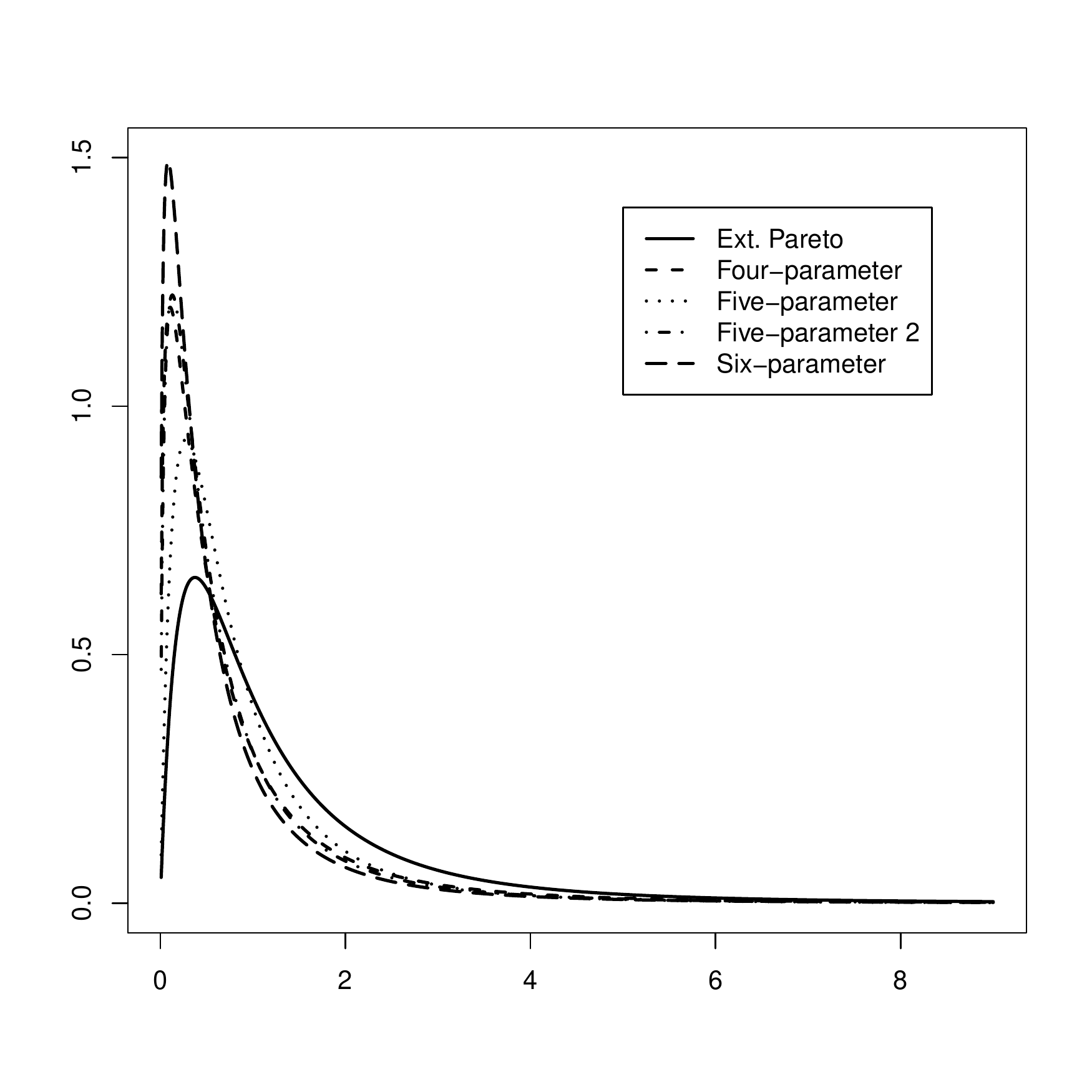}
\caption{Pdfs of the $10$ distributions considered in the simulation study.}
\label{fig:simstudypdfs}
\end{figure}

To assess the effect of the sample size, we have run the experiments for 
each of $n = 5,000, \ 500, \ 50$, which corresponds to data sets ranging 
from rather large to rather small, but realistic, depending on the line of 
business. Further, we used $m=100,000$ in the Monte Carlo estimation of the 
reserve. Finally, we let $N=1,000$ in all experiments.

\subsection{Parameter estimation}\label{subsec:estimation}

All the distribution parameters are estimated by maximum likelihood
estimation, using the R function \texttt{optim()} to do 
quasi-Newton optimisation. For the two-parameter families, this
optimisation is straightforward. However, especially the five-
and six-parameter distributions appear to have likelihood functions
that are rather flat in some of the parameters. This means that
widely different parameter values result in rather similar 
distributions. Consequently, the optimisation of the likelihood
is sometimes challenging, but this is not that problematic when
the resulting distributions are sensible. Still, we have derived 
the derivatives of the log-likelihood functions with respect to
all parameters, and supply those to \texttt{optim()}. These are
given in the Appendix. To ease the optimisation further, it is 
performed on the log-transformed parameters as all the parameters 
$\alpha,\theta,\beta,\eta,\tau,\gamma$ are positive. Finally,
we do the optimisation with several sets of start values for the 
parameters, and choose the parameter estimates that give the
highest likelihood value. More specifically, we use the moment
estimates $(\tilde{\alpha}^{ep},\tilde{\theta}^{ep},\tilde{\beta}^{ep})$,
as well as $(10^{10},\hat{\alpha}^{ga},\hat{\xi}^{ga})$ and 
$(\hat{\alpha}^{pa},1,\hat{\beta}^{pa})$ for the extended Pareto
distribution, where $(\hat{\xi}^{ga},\hat{\alpha}^{ga})$ and 
$(\hat{\alpha}^{pa},\hat{\beta}^{pa})$ are the maximum likelihood
estimates obtained for the Gamma and Pareto distributions, respectively.
For the four-parameter distribution, we employ 
$(\hat{\alpha}^{ep},\hat{\theta}^{ep},\hat{\beta}^{ep},1)$
and $(10^{10},1,\hat{\eta}^{we},1/\hat{\beta}^{we})$, where
$(\hat{\alpha}^{ep},\hat{\theta}^{ep},\hat{\beta}^{ep})$
and $(\hat{\beta}^{we},\hat{\eta}^{we})$ are the estimates
from the extended Pareto and Weibull distributions. The
start values used for the five-parameter distribution are 
$(\hat{\alpha}^{ep},\hat{\theta}^{ep},\hat{\beta}^{ep},1,1)$ 
and $(10^{10},\hat{\theta}^{lga},1,10^{10},\hat{\xi}^{lga}\cdot 10^{10})$,
where $(\hat{\xi}^{lga},\hat{\theta}^{lga})$ is the estimate
from the log-Gamma distribution. For the five-parameter 2
distribution, the start values are 
$(\hat{\alpha}^{fp},\hat{\theta}^{fp},\hat{\beta}^{fp},\hat{\eta}^{fp},1)$,
where $(\hat{\alpha}^{fp},\hat{\theta}^{fp},\hat{\beta}^{fp},\hat{\eta}^{fp})$
are the estimates from the four-parameter distribution. Finally,
the start values for the six-parameter distribution are\\
$(\hat{\alpha}^{fip},\hat{\theta}^{fip},\hat{\beta}^{fip},1,\hat{\tau}^{fip},\hat{\gamma}^{fip})$
and 
$(\hat{\alpha}^{fip2},\hat{\theta}^{fip2},\hat{\beta}^{fip2},\hat{\eta}^{fip2},1,\hat{\gamma}^{fip2})$,
where\\ 
$(\hat{\alpha}^{fip},\hat{\theta}^{fip},\hat{\beta}^{fip},\hat{\tau}^{fip},\hat{\gamma}^{fip})$
and
$(\hat{\alpha}^{fip2},\hat{\theta}^{fip2},\hat{\beta}^{fip2},\hat{\eta}^{fip2},\hat{\gamma}^{fip2})$
are the estimates from the five-parameter and five-parameter 2 
distributions.

\subsection{Results}\label{subsec:results}

The results from the simulations in terms of the bias and root
mean squared error (RMSE) of the quantile and reserve estimates 
are shown in Tables \ref{tab:res.z.0.95} to 
\ref{tab:res.x.0.99.1000.l}. The relative performance of the models 
with $n=500$ is rather similar to that with $n=5,000$, though with 
larger biases and RMSEs. Hence, the results from the corresponding
simulations are not shown here. Furthermore, our main interest is the
performance of the six-parameter distribution and its five- and 
four-parameter versions when the true distribution is one of the 
six special cases from Table \ref{tab:speccases}, i.e. whether our
model is able to capture the behaviour of these different 
distributions in practice for finite sample sizes, as it has been
shown to do in theory. Hence, we only report the results from the
simulations with one of the Table \ref{tab:speccases} distributions
as the true distribution. 

Tables \ref{tab:res.z.0.95} and \ref{tab:res.z.0.99} display the 
bias and RMSE in the estimates of $q_{\epsilon}^{Z}$ for 
$\epsilon=0.05$ and $0.01$, respectively, when the sample size is
$5,000$. Corresponding results for sample size $50$ are shown in 
Tables \ref{tab:res.z.0.95.l} and \ref{tab:res.z.0.99.l}. We see 
that when the sample size is large, the quantile estimates from 
the six-parameter, five-parameter 2 and four-parameter distribution 
are just as good as the ones from the true distribution, both in terms 
of bias and RMSE. The performance of the five-parameter distribution
is good for the $95\%$ quantile, but inferior to the others for the
$99\%$ quantile. As mentioned above, the impression from the simulations
with $n=500$ are the same. Hence, our model performs well also for
medium sample sizes. When the sample size is reduced to $50$, we see 
that the estimates of $q_{0.05}^{Z}$ from our model are still rather
good in terms of bias, but with comparatively higher RMSE, especially
for the most heavy-tailed distributions. As one moves further out in 
the tail of the claim size distribution, the quality of the estimates
is reasonable except when the true distribution is log-normal or
log-Gamma. As mentioned in Section \ref{sec:model}, the main source
of error in the reserve estimates from the six-parameter distribution
is likely to be the estimation error, while the Monte Carlo and model
error should be small. Hence, the decreased performance is probably due 
to increased parameter uncertainty for the smaller sample size.
 
Tables \ref{tab:res.x.0.95.10} and \ref{tab:res.x.0.99.10} show the 
bias and RMSE in the estimates of the reserve $q_{\epsilon}$ for 
$\epsilon=0.05$ and $0.01$, respectively, when the parameter of the
claim frequency distribution is $\lambda=10$ and $n=5,000$. 
Corresponding results for $\lambda=1,000$ are displayed in Tables
\ref{tab:res.x.0.95.1000} and \ref{tab:res.x.0.99.1000}. We see that
the reserve estimates from our model are very good, at least for all
distributions but the log-normal. The six- and fiv-parameter 2 
distributions perform rather well also for the log-normal claims. The
four- and five-parameter distributions do not perform that well in 
this case, especially for the $99\%$ reserve. This may be due to the 
fact that these two distributions do not have the log-normal as a 
special case. When $\lambda$ increases to $1,000$, the relative 
performance of the applied models is similar, though the difference
between them is smaller, as one would expect when one sums many
independent, identically distributed variables. The results for
$\lambda=100$ are somewhere between those for $\lambda=10$ and
$\lambda=1,000$, and are therefore not shown here.

Results corresponding to the ones in Tables \ref{tab:res.x.0.95.10}
to \ref{tab:res.x.0.99.1000} for $n=50$ are displayed in Tables 
\ref{tab:res.x.0.95.10.l} to \ref{tab:res.x.0.99.1000.l}. The five-
and six-parameter versions of our model still produce good estimates
of the $95\%$ reserve when $\lambda=10$. The quality of the $99\%$
reserve estimates is also relatively good, except for the log-normal
and log-Gamma distributed claims. The four-parameter distribution
does not perform well in this case. When $\lambda$ increases to
$1,000$, our model still gives reasonable estimates of the $95$
and $99\%$ reserves of the medium-tailed Weibull and Gamma claims.
However, it does not perform well for claims from heavy-tailed
distributions. 

Thus, all in all, our model appears to capture the tail behaviour
of claims from widely different distributions, ranging from the
moderate-tailed Weibull and Gamma to the very heavy-tailed Pareto
and extended Pareto when the amount of data is large or medium. 
When the sample size is small, however, our model should be used
with care. For medium-tailed claim sizes, the performance seems
to be reasonable, but not for heavy-tailed ones. Very large RMSEs
in such cases indicate that the parameter uncertainty is huge. 
Further, the performance of the five-parameter 2 and six-parameter
distributions is consistently better than that of the four- and 
five-parameter ones, even for small sample sizes. Also, the 
five-parameter 2 gives very similar results to the six-parameter.
This indicates that the sixth parameter $\tau$ may not be needed,
but that the fifth parameter $\gamma$ makes a significant 
contribution to the model.

\begin{table}[t]
	\begin{center}
		\begin{tabular}{lcccccc}
			\hline
			\hline
			& \multicolumn{6}{c}{\textsc{Bias}}\\
			\textsc{A$\backslash$T} & L-N & L-G & We & Pa & Ga & E. Pa.\\
			\hline
			L-N &  0.004 &  0.023 &  0.476 &  1.250 &  0.487 &  0.338\\
			L-G & -0.164 &  0.000 &  0.385 &  0.309 &  0.341 &  0.176\\
			We & -0.143 & -0.056 &  0.000 & -0.016 & -0.047 & -0.041\\
			Pa &  0.043 &  0.809 &  1.045 &  0.001 &  0.624 &  0.205\\
			Ga & -0.275 & -0.167 &  0.119 & -0.096 &  0.000 & -0.237\\
			E. Pa. &  0.004 & -0.001 &  0.120 &  0.002 &  0.001 &  0.004\\
			4-par. &  0.008 & -0.001 &  0.001 &  0.001 & -0.001 &  0.004\\
			5-par. &  0.007 & -0.001 &  0.039 &  0.001 &  0.000 &  0.003\\
			5-par. 2 & -0.009 & -0.001 &  0.001 &  0.002 & -0.001 &  0.004\\
			6-par & -0.004 & -0.001 &  0.001 &  0.000 & -0.001 &  0.003\\
			\hline
			& \multicolumn{6}{c}{\textsc{RMSE}}\\
			\textsc{A$\backslash$T} & L-N & L-G & We & Pa & Ga & E. Pa.\\
			\hline
			L-N & 0.071 & 0.049 & 0.477 & 1.258 & 0.489 & 0.356\\
			L-G & 0.177 & 0.043 & 0.386 & 0.324 & 0.343 & 0.207\\
			We & 0.162 & 0.077 & 0.015 & 0.097 & 0.055 & 0.139\\
			Pa & 0.084 & 0.810 & 1.045 & 0.088 & 0.625 & 0.232\\
			Ga & 0.284 & 0.172 & 0.121 & 0.137 & 0.028 & 0.265\\
			E. Pa. & 0.080 & 0.047 & 0.121 & 0.089 & 0.029 & 0.112\\
			4-par. & 0.077 & 0.047 & 0.016 & 0.091 & 0.030 & 0.113\\
			5-par. & 0.080 & 0.047 & 0.057 & 0.090 & 0.029 & 0.112\\
			5-par. 2 & 0.077 & 0.047 & 0.016 & 0.090 & 0.030 & 0.112\\
			6-par & 0.077 & 0.047 & 0.016 & 0.091 & 0.029 & 0.112\\
			\hline
			\hline
		\end{tabular}
		\caption{Bias and RMSE on the upper $95\%$ quantile of the distribution of $Z$ from each of the applied (A) distributions and each distribution from Table \ref{tab:speccases} as the true (T) one when $n=5,000$.}
		\label{tab:res.z.0.95}
	\end{center}
\end{table}

\begin{table}[t]
	\begin{center}
		\begin{tabular}{lcccccc}
			\hline
			\hline
			& \multicolumn{6}{c}{\textsc{Bias}}\\
			\textsc{A$\backslash$T} & L-N & L-G & We & Pa & Ga & E. Pa.\\
			\hline
			L-N &  0.010 & -0.060 &  1.340 &  5.105 &  1.622 &  0.615\\
			L-G & -0.305 &  0.000 &  1.185 &  2.131 &  1.362 &  0.911\\
			We & -1.600 & -0.862 &  0.000 & -1.453 & -0.212 & -2.243\\
			Pa & -0.664 &  1.073 &  2.188 &  0.007 &  1.287 & -0.982\\
			Ga & -1.886 & -0.911 &  0.333 & -1.914 & -0.001 & -2.667\\
			E. Pa. &  0.735 &  0.013 &  0.333 &  0.008 &  0.007 &  0.020\\
			4-par. &  0.296 &  0.013 &  0.004 &  0.009 & -0.001 &  0.021\\
			5-par. &  0.668 &  0.008 &  0.105 &  0.018 &  0.005 &  0.014\\
			5-par. 2 &  0.081 &  0.009 &  0.003 &  0.007 & -0.001 &  0.009\\
			6-par &  0.028 &  0.005 &  0.003 &  0.006 &  0.000 &  0.006\\
			\hline
			& \multicolumn{6}{c}{\textsc{RMSE}}\\
			\textsc{A$\backslash$T} & L-N & L-G & We & Pa & Ga & E. Pa.\\
			\hline
			L-N & 0.175 & 0.104 & 1.342 & 5.132 & 1.625 & 0.680\\
			L-G & 0.354 & 0.091 & 1.186 & 2.161 & 1.366 & 0.969\\
			We & 1.606 & 0.867 & 0.022 & 1.467 & 0.216 & 2.258\\
			Pa & 0.692 & 1.075 & 2.188 & 0.320 & 1.288 & 1.025\\
			Ga & 1.890 & 0.913 & 0.334 & 1.922 & 0.044 & 2.674\\
			E. Pa. & 0.797 & 0.126 & 0.334 & 0.325 & 0.051 & 0.384\\
			4-par. & 0.393 & 0.127 & 0.026 & 0.338 & 0.057 & 0.399\\
			5-par. & 0.744 & 0.123 & 0.154 & 0.336 & 0.052 & 0.392\\
			5-par. 2 & 0.279 & 0.127 & 0.026 & 0.339 & 0.057 & 0.403\\
			6-par & 0.268 & 0.126 & 0.026 & 0.342 & 0.055 & 0.406\\
			\hline
			\hline
		\end{tabular}
		\caption{Bias and RMSE on the upper $99\%$ quantile of the distribution of $Z$ from each of the applied (A) distributions and each distribution from Table \ref{tab:speccases} as the true (T) one when $n=5,000$.}
		\label{tab:res.z.0.99}
\end{center}
\end{table}

\begin{table}[t]
	\begin{center}
		\begin{tabular}{lcccccc}
			\hline
			\hline
			& \multicolumn{6}{c}{\textsc{Bias}}\\
			\textsc{A$\backslash$T} & L-N & L-G & We & Pa & Ga & E. Pa.\\
			\hline
			L-N &  0.113 &  0.049 &  0.462 &  1.312 &  0.501 &  0.316\\
			L-G & -0.094 & -0.005 &  0.355 &  0.276 &  0.324 &  0.103\\
			We & -0.111 & -0.100 & -0.017 & -0.120 & -0.071 & -0.242\\
			Pa &  0.102 &  0.825 &  1.042 & -0.050 &  0.629 &  0.084\\
			Ga & -0.210 & -0.172 &  0.100 & -0.170 & -0.016 & -0.349\\
			E. Pa. &  0.145 & -0.002 &  0.102 & -0.034 &  0.010 & -0.075\\
			4-par. &  0.123 & -0.009 & -0.020 & -0.043 & -0.041 & -0.077\\
			5-par. &  0.079 & -0.029 & -0.003 & -0.035 & -0.031 & -0.099\\
			5-par. 2 &  0.066 & -0.018 & -0.021 & -0.102 & -0.042 & -0.125\\
			6-par &  0.059 & -0.016 & -0.018 & -0.078 & -0.038 & -0.131\\
			\hline
			& \multicolumn{6}{c}{\textsc{RMSE}}\\
			\textsc{A$\backslash$T} & L-N & L-G & We & Pa & Ga & E. Pa.\\
			\hline
			L-N & 0.719 & 0.448 & 0.551 & 1.962 & 0.666 & 1.131\\
			L-G & 0.684 & 0.439 & 0.435 & 1.005 & 0.494 & 1.060\\
			We & 0.722 & 0.507 & 0.155 & 0.881 & 0.278 & 1.111\\
			Pa & 0.688 & 0.925 & 1.064 & 0.898 & 0.694 & 0.991\\
			Ga & 0.730 & 0.463 & 0.196 & 0.923 & 0.266 & 1.122\\
			E. Pa. & 0.861 & 0.532 & 0.197 & 0.927 & 0.284 & 1.168\\
			4-par. & 0.827 & 0.529 & 0.160 & 0.949 & 0.290 & 1.148\\
			5-par. & 0.785 & 0.497 & 0.160 & 0.933 & 0.285 & 1.111\\
			5-par. 2 & 0.787 & 0.522 & 0.160 & 0.921 & 0.289 & 1.121\\
			6-par & 0.779 & 0.515 & 0.161 & 0.933 & 0.289 & 1.112\\
			\hline
			\hline
		\end{tabular}
		\caption{Bias and RMSE on the upper $95\%$ quantile of the distribution of $Z$ from each of the applied (A) distributions and each distribution from Table \ref{tab:speccases} as the true (T) one when $n=50$.}
	\label{tab:res.z.0.95.l}
\end{center}
\end{table}

\begin{table}[t]
	\begin{center}
		\begin{tabular}{lcccccc}
			\hline
			\hline
			& \multicolumn{6}{c}{\textsc{Bias}}\\
			\textsc{A$\backslash$T} & L-N & L-G & We & Pa & Ga & E. Pa.\\
			\hline
			L-N &  0.325 &  0.007 &  1.329 &  5.675 &  1.681 &  0.643\\
			L-G & -0.052 &  0.007 &  1.132 &  2.215 &  1.340 &  0.825\\
			We & -1.536 & -0.933 & -0.026 & -1.652 & -0.253 & -2.610\\
			Pa & -0.433 &  1.117 &  2.183 &  0.127 &  1.296 & -1.180\\
			Ga & -1.776 & -0.918 &  0.300 & -2.042 & -0.030 & -2.852\\
			E. Pa. &  1.725 &  0.140 &  0.310 &  0.394 &  0.112 &  0.231\\
			4-par. &  1.452 &  0.266 & -0.021 &  0.594 & -0.016 &  0.445\\
			5-par. &  0.691 & -0.006 &  0.030 &  0.280 & -0.014 & -0.187\\
			5-par. 2 &  0.676 &  0.132 & -0.021 & -0.160 & -0.025 & -0.230\\
			6-par &  0.318 &  0.094 & -0.013 & -0.124 & -0.023 & -0.421\\
			\hline
			& \multicolumn{6}{c}{\textsc{RMSE}}\\
			\textsc{A$\backslash$T} & L-N & L-G & We & Pa & Ga & E. Pa.\\
			\hline
			L-N & 1.841 & 0.893 & 1.491 & 8.173 & 1.993 & 2.947\\
			L-G & 1.891 & 0.939 & 1.261 & 4.298 & 1.620 & 3.392\\
			We & 2.032 & 1.235 & 0.224 & 2.443 & 0.492 & 3.317\\
			Pa & 1.941 & 1.316 & 2.209 & 3.612 & 1.371 & 2.951\\
			Ga & 2.130 & 1.137 & 0.397 & 2.595 & 0.417 & 3.377\\
			E. Pa. & 4.173 & 1.615 & 0.410 & 4.359 & 0.668 & 4.901\\
			4-par. & 4.042 & 1.752 & 0.274 & 5.033 & 0.660 & 5.019\\
			5-par. & 2.762 & 1.341 & 0.271 & 3.900 & 0.583 & 3.893\\
			5-par. 2 & 2.809 & 1.494 & 0.274 & 3.444 & 0.638 & 3.935\\
			6-par & 2.523 & 1.445 & 0.279 & 3.592 & 0.606 & 3.701\\
			\hline
			\hline
		\end{tabular}
	\caption{Bias and RMSE on the upper $99\%$ quantile of the distribution of $Z$ from each of the applied (A) distributions and each distribution from Table \ref{tab:speccases} as the true (T) one when $n=50$.}
		\label{tab:res.z.0.99.l}
\end{center}
\end{table}

\begin{table}[t]
	\begin{center}
		\begin{tabular}{lcccccc}
			\hline
			\hline
			& \multicolumn{6}{c}{\textsc{Bias}}\\
			\textsc{A$\backslash$T} & L-N & L-G & We & Pa & Ga & E. Pa.\\
			\hline
			L-N &  0.025 & -0.054 &  1.344 &  9.265 &  2.012 &  0.976\\
			L-G & -0.399 &  0.001 &  1.073 &  3.835 &  1.568 &  1.496\\
			We & -1.437 & -0.215 &  0.000 & -1.875 & -0.041 & -2.104\\
			Pa & -0.662 &  1.315 &  1.853 &  0.006 &  1.224 & -1.136\\
			Ga & -1.719 & -0.549 &  0.137 & -2.064 & -0.004 & -2.520\\
			E. Pa. &  1.321 &  0.018 &  0.138 &  0.009 & -0.001 &  0.038\\
			4-par. &  0.523 &  0.018 &  0.000 &  0.024 & -0.006 &  0.046\\
			5-par. &  1.191 &  0.012 &  0.043 &  0.028 & -0.005 &  0.034\\
			5-par. 2 &  0.193 &  0.012 &  0.000 &  0.006 & -0.005 &  0.026\\
			6-par &  0.097 &  0.006 & -0.001 &  0.009 & -0.005 &  0.029\\
			\hline
			& \multicolumn{6}{c}{\textsc{RMSE}}\\
			\textsc{A$\backslash$T} & L-N & L-G & We & Pa & Ga & E. Pa.\\
			\hline
			L-N & 0.419 & 0.250 & 1.353 & 9.329 & 2.027 & 1.181\\
			L-G & 0.587 & 0.253 & 1.082 & 3.917 & 1.585 & 1.672\\
			We & 1.487 & 0.334 & 0.120 & 1.934 & 0.182 & 2.194\\
			Pa & 0.789 & 1.337 & 1.858 & 0.652 & 1.239 & 1.302\\
			Ga & 1.760 & 0.599 & 0.184 & 2.121 & 0.177 & 2.593\\
			E. Pa. & 1.461 & 0.273 & 0.184 & 0.655 & 0.178 & 0.781\\
			4-par. & 0.748 & 0.269 & 0.123 & 0.676 & 0.179 & 0.798\\
			5-par. & 1.358 & 0.270 & 0.138 & 0.676 & 0.179 & 0.795\\
			5-par. 2 & 0.562 & 0.273 & 0.121 & 0.678 & 0.179 & 0.808\\
			6-par & 0.540 & 0.273 & 0.121 & 0.682 & 0.180 & 0.817\\
			\hline
			\hline
		\end{tabular}
		\caption{Bias and RMSE on the upper $95\%$ reserve from each of the applied (A) distributions and each distribution from Table \ref{tab:speccases} as the true (T) one when $\lambda=10$ and $n=5,000$.}
	\label{tab:res.x.0.95.10}
\end{center}
\end{table}

\begin{table}[t]
	\begin{center}
		\begin{tabular}{lcccccc}
			\hline
			\hline
			& \multicolumn{6}{c}{\textsc{Bias}}\\
			\textsc{A$\backslash$T} & L-N & L-G & We & Pa & Ga & E. Pa.\\
			\hline
			L-N &  0.054 & -0.199 &  2.112 & 20.726 &  3.478 &  0.423\\
			L-G & -0.079 &  0.006 &  1.760 & 12.505 &  2.967 &  4.066\\
			We & -3.958 & -0.781 &  0.008 & -6.153 & -0.115 & -7.145\\
			Pa & -2.140 &  1.936 &  3.162 &  0.042 &  2.089 & -4.802\\
			Ga & -4.444 & -1.221 &  0.252 & -6.763 & -0.004 & -7.870\\
			E. Pa. &  5.122 &  0.054 &  0.253 &  0.039 &  0.007 &  0.133\\
			4-par. &  1.809 &  0.057 &  0.009 &  0.095 & -0.005 &  0.169\\
			5-par. &  4.536 &  0.039 &  0.082 &  0.135 & -0.004 &  0.125\\
			5-par. 2 &  0.814 &  0.050 &  0.004 &  0.057 &  0.000 &  0.091\\
			6-par &  0.459 &  0.033 &  0.006 &  0.076 &  0.001 &  0.106\\
			\hline
			& \multicolumn{6}{c}{\textsc{RMSE}}\\
			\textsc{A$\backslash$T} & L-N & L-G & We & Pa & Ga & E. Pa.\\
			\hline
			L-N &  0.663 &  0.380 &  2.122 & 20.867 &  3.496 &  1.166\\
			L-G &  0.764 &  0.346 &  1.771 & 12.669 &  2.988 &  4.335\\
			We &  3.991 &  0.846 &  0.150 &  6.189 &  0.249 &  7.194\\
			Pa &  2.241 &  1.962 &  3.168 &  1.464 &  2.103 &  4.903\\
			Ga &  4.471 &  1.258 &  0.295 &  6.793 &  0.225 &  7.909\\
			E. Pa. &  5.363 &  0.405 &  0.296 &  1.523 &  0.228 &  1.679\\
			4-par. &  2.142 &  0.404 &  0.150 &  1.736 &  0.225 &  1.867\\
			5-par. &  4.871 &  0.398 &  0.194 &  1.761 &  0.223 &  1.763\\
			5-par. 2 &  1.388 &  0.412 &  0.150 &  1.764 &  0.225 &  1.909\\
			6-par &  1.232 &  0.407 &  0.150 &  1.776 &  0.226 &  1.964\\
			\hline
			\hline
		\end{tabular}
		\caption{Bias and RMSE on the upper $99\%$ reserve from each of the applied (A) distributions and each distribution from Table \ref{tab:speccases} as the true (T) one when $\lambda=10$ and $n=5,000$.}
	\label{tab:res.x.0.99.10}
\end{center}
\end{table}

\begin{table}[t]
	\begin{center}
		\begin{tabular}{lcccccc}
			\hline
			\hline
			& \multicolumn{6}{c}{\textsc{Bias}}\\
			\textsc{A$\backslash$T} & L-N & L-G & We & Pa & Ga & E. Pa.\\
			\hline
			L-N & 1.26 & -0.49 & 46.57 & 319.30 & 66.01 & 29.99\\
			L-G & -9.65 & 0.02 & 34.58 & 161.65 & 47.61 & 53.50\\
			We & -12.18 & 8.62 & 0.03 & -39.49 & 2.79 & -26.221\\
			Pa & -11.98 & 11.10 & 15.36 & 0.99 & 10.32 & -31.01\\
			Ga & -14.39 & -4.37 & 1.10 & -25.48 & 0.14 & -29.39\\
			E. Pa. & 55.84 & 0.47 & 1.02 & 1.17 & 0.08 & 1.76\\
			4-par. & 19.15 & 0.52 & -0.11 & 1.91 & 0.07 & 2.32\\
			5-par. & 48.58 & 0.31 & 0.49 & 2.32 & -0.01 & 1.63\\
			5-par. 2 & 8.61 & 0.39 & -0.13 & 1.66 & 0.07 & 1.58\\
			6-par & 5.25 & 0.23 & -0.14 & 1.80 & 0.00 & 1.49\\
			\hline
			& \multicolumn{6}{c}{\textsc{RMSE}}\\
			\textsc{A$\backslash$T} & L-N & L-G & We & Pa & Ga & E. Pa.\\
			\hline
			L-N & 20.12 & 13.65 & 47.34 & 322.21 & 67.19 & 42.76\\
			L-G & 23.02 & 13.85 & 35.57 & 166.01 & 49.17 & 63.76\\
			We & 23.60 & 16.73 & 7.84 & 45.75 & 11.27 & 40.92\\
			Pa & 23.50 & 17.85 & 17.32 & 27.93 & 15.07 & 43.01\\
			Ga & 25.362 & 14.62 & 7.93 & 36.20 & 10.89 & 44.28\\
			E. Pa. & 62.72 & 14.10 & 7.93 & 28.00 & 10.91 & 34.87\\
			4-par. & 30.11 & 14.10 & 7.96 & 29.30 & 10.92 & 35.83\\
			5-par. & 57.06 & 14.05 & 8.13 & 29.40 & 10.92 & 35.25\\
			5-par. 2 & 24.55 & 14.12 & 7.98 & 29.37 & 10.93 & 36.04\\
			6-par & 23.68 & 14.11 & 8.00 & 29.54 & 10.96 & 36.09\\
			\hline
			\hline
		\end{tabular}
		\caption{Bias and RMSE on the upper $95\%$ reserve from each of the applied (A) distributions and each distribution from Table \ref{tab:speccases} as the true (T) one when $\lambda=1000$ and $n=5,000$.}
	\label{tab:res.x.0.95.1000}
\end{center}
\end{table}

\begin{table}[t]
	\begin{center}
		\begin{tabular}{lcccccc}
			\hline
			\hline
			& \multicolumn{6}{c}{\textsc{Bias}}\\
			\textsc{A$\backslash$T} & L-N & L-G & We & Pa & Ga & E. Pa.\\
			\hline
			L-N & 1.29 & -0.90 & 50.14 & 372.22 & 71.81 & 23.28\\
			L-G & -6.77 & -0.07 & 37.54 & 266.02 & 52.73 & 70.31\\
			We & -19.53 & 7.27 & 0.02 & -58.29 & 2.57 & -46.76\\
			Pa & -15.90 & 15.76 & 22.18 & 1.36 & 14.87 & -47.18\\
			Ga & -22.77 & -6.56 & 1.61 & -45.57 & 0.12 & -51.52\\
			E. Pa. & 88.52 & 0.49 & 1.51 & 1.70 & 0.07 & 2.49\\
			4-par. & 25.86 & 0.51 & -0.09 & 3.53 & 0.06 & 3.70\\
			5-par. & 75.18 & 0.30 & 0.67 & 4.06 & 0.01 & 2.46\\
			5-par. 2 & 12.03 & 0.40 & -0.11 & 3.15 & 0.06 & 2.82\\
			6-par & 7.43 & 0.20 & -0.14 & 3.58 & 0.00 & 2.66\\
			\hline
			& \multicolumn{6}{c}{\textsc{RMSE}}\\
			\textsc{A$\backslash$T} & L-N & L-G & We & Pa & Ga & E. Pa.\\
			\hline
			L-N & 21.14 & 14.11 & 50.90 & 375.62 & 72.98 & 39.77\\
			L-G & 23.55 & 14.34 & 38.50 & 271.57 & 54.25 & 80.64\\
			We & 28.66 & 16.48 & 8.03 & 63.08 & 11.49 & 57.03\\
			Pa & 26.54 & 21.35 & 23.68 & 32.74 & 18.69 & 56.60\\
			Ga & 31.38 & 15.77 & 8.19 & 52.80 & 11.21 & 61.88\\
			E. Pa. & 96.18 & 14.64 & 8.20 & 33.14 & 11.19 & 39.86\\
			4-par. & 36.74 & 14.66 & 8.16 & 36.89 & 11.19 & 42.55\\
			5-par. & 85.49 & 14.58 & 8.33 & 37.11 & 11.22 & 40.86\\
			5-par. 2 & 28.17 & 14.65 & 8.17 & 37.06 & 11.21 & 43.02\\
			6-par & 26.62 & 14.65 & 8.18 & 37.48 & 11.22 & 43.22\\
			\hline
			\hline
		\end{tabular}
		\caption{Bias and RMSE on the upper $99\%$ reserve from each of the applied (A) distributions and each distribution from Table \ref{tab:speccases} as the true (T) one when $\lambda=1000$ and $n=5,000$.}
	\label{tab:res.x.0.99.1000}
\end{center}
\end{table}

\begin{table}[t]
	\begin{center}
		\begin{tabular}{lcccccc}
			\hline
			\hline
			& \multicolumn{6}{c}{\textsc{Bias}}\\
			\textsc{A$\backslash$T} & L-N & L-G & We & Pa & Ga & E. Pa.\\
			\hline
			L-N & 0.79 & 0.14 & 1.37 & 10.56 & 2.22 & 1.08\\
			L-G & 0.23 & 0.08 & 1.02 & 4.14 & 1.60 & 1.40\\
			We & -1.13 & -0.23 & -0.03 & -2.24 & -0.05 & -2.79\\
			Pa & -0.05 & 1.40 & 1.84 & 0.54 & 1.26 & -1.48\\
			Ga & -1.32 & -0.50 & 0.10 & -2.37 & -0.00 & -2.98\\
			E. Pa. & 3.66 & 0.40 & 0.11 & 1.20 & 0.13 & 0.97\\
			4-par. & 3.10 & 0.56 & -0.02 & 1.64 & 0.06 & 1.21\\
			5-par. & 1.59 & 0.15 & -0.02 & 0.88 & 0.04 & -0.02\\
			5-par. 2 & 1.46 & 0.31 & -0.02 & -0.11 & 0.05 & -0.20\\
			6-par & 0.90 & 0.26 & -0.01 &  0.06 & 0.05 & -0.52\\
			\hline
			& \multicolumn{6}{c}{\textsc{RMSE}}\\
			\textsc{A$\backslash$T} & L-N & L-G & We & Pa & Ga & E. Pa.\\
			\hline
			L-N & 4.31 & 2.51 & 2.03 & 16.10 & 3.24 & 6.59\\
			L-G & 4.44 & 2.56 & 1.72 & 9.17 & 2.71 & 7.52\\
			We & 3.81 & 2.57 & 1.15 & 5.01 & 1.65 & 6.10\\
			Pa & 4.29 & 2.87 & 2.26 & 7.67 & 2.17 & 6.36\\
			Ga & 3.90 & 2.50 & 1.17 & 5.21 & 1.66 & 6.29\\
			E. Pa. & 8.79 & 3.33 & 1.17 & 9.53 & 1.80 & 10.08\\
			4-par. & 8.71 & 3.51 & 1.15 & 11.56 & 1.72 & 10.66\\
			5-par. & 5.95 & 2.96 & 1.16 & 8.34 & 1.69 & 8.26\\
			5-par. 2 & 5.96 & 3.11 & 1.15 & 7.22 & 1.71 & 8.24\\
			6-par & 5.41 & 3.08 & 1.16 & 7.64 & 1.69 & 7.77\\
			\hline
			\hline
		\end{tabular}
		\caption{Bias and RMSE on the upper $95\%$ reserve from each of the applied (A) distributions and each distribution from Table \ref{tab:speccases} as the true (T) one when $\lambda=10$ and $n=50$.}
	\label{tab:res.x.0.95.10.l}
\end{center}
\end{table}

\begin{table}[t]
	\begin{center}
		\begin{tabular}{lcccccc}
			\hline
			\hline
			& \multicolumn{6}{c}{\textsc{Bias}}\\
			\textsc{A$\backslash$T} & L-N & L-G & We & Pa & Ga & E. Pa.\\
			\hline
			L-N & 1.41 & 0.10 & 2.18 & 25.13 & 3.85 & 0.89\\
			L-G & 1.38 & 0.16 & 1.70 & 14.84 & 3.09 & 4.68\\
			We & -3.56 & -0.83 & -0.04 & -6.63 & -0.14 & -8.07\\
			Pa & -0.66 & 2.07 & 3.14 & 3.86 & 2.13 & -4.64\\
			Ga & -3.94 & -1.17 & 0.20 & -7.16 & -0.01 & -8.46\\
			E. Pa. & 14.70 & 1.25 & 0.21 & 7.89 & 0.32 & 6.33\\
			4-par. & 12.82 & 2.03 & -0.02 & 11.57 & 0.20 & 8.24\\
			5-par. & 6.07 & 0.61 & -0.010 & 5.89 & 0.09 & 2.45\\
			5-par. 2 & 5.53 & 1.09 & -0.02 & 2.37 & 0.16 & 2.05\\
			6-par & 3.82 & 1.08 & -0.01 & 2.93 & 0.14 & 1.02\\
			\hline
			& \multicolumn{6}{c}{\textsc{RMSE}}\\
			\textsc{A$\backslash$T} & L-N & L-G & We & Pa & Ga & E. Pa.\\
			\hline
			L-N & 6.97 & 3.29 & 2.99 & 39.46 & 5.22 & 10.85\\
			L-G & 8.32 & 3.50 & 2.51 & 27.20 & 4.50 & 16.41\\
			We & 5.95 & 3.31 & 1.36 & 9.05 & 2.01 & 10.80\\
			Pa & 7.67 & 3.79 & 3.54 & 24.29 & 3.06 & 11.92\\
			Ga & 6.13 & 3.24 & 1.40 & 9.37 & 2.02 & 11.06\\
			E. Pa. & 30.94 & 6.36 & 1.41 & 42.30 & 2.71 & 32.95\\
			4-par. & 35.87 & 8.03 & 1.36 & 62.22 & 2.37 & 39.93\\
			5-par. & 18.65 & 4.92 & 1.37 & 31.50 & 2.13 & 25.01\\
			5-par. 2 & 17.54 & 5.56 & 1.36 & 23.48 & 2.35 & 23.98\\
			6-par & 16.97 & 6.25 & 1.37 & 27.72 & 2.28 & 22.07\\
			\hline
			\hline
		\end{tabular}
		\caption{Bias and RMSE on the upper $99\%$ reserve from each of the applied (A) distributions and each distribution from Table \ref{tab:speccases} as the true (T) one when $\lambda=10$ and $n=50$.}
	\label{tab:res.x.0.99.10.l}
\end{center}
\end{table}

\begin{table}[t]
	\begin{center}
		\begin{tabular}{lcccccc}
			\hline
			\hline
			& \multicolumn{6}{c}{\textsc{Bias}}\\
			\textsc{A$\backslash$T} & L-N & L-G & We & Pa & Ga & E. Pa.\\
			\hline
			L-N & 35.95 & 9.37 & 47.76 & 379.00 & 75.69 & 33.02\\
			L-G & 21.65 & 4.37 & 32.34 & 203.38 & 50.69 & 55.15\\
			We & 6.11 & 10.80 & -0.91 & -55.01 & 3.86 & -56.26\\
			Pa & 16.91 & 14.76 & 14.28 & 63.59 & 12.09 & -47.55\\
			Ga & 8.57 & -0.46 & -0.15 & -40.20 & 1.65 & -52.97\\
			E. Pa. & 230.25 & 16.97 & 0.02 & 164.58 & 4.95 & 100.62\\
			4-par. & 221.97 & 25.83 & -0.47 & 279.04 & 4.26 & 143.03\\
			5-par. & 107.65 & 7.68 & -1.03 & 113.62 & 2.71 & 82.51\\
			5-par. 2 & 88.75 & 12.74 & -0.48 & 57.37 & 3.86 & 38.23\\
			6-par & 97.79 & 18.13 & -0.25 & 112.87 & 3.84 & 31.50\\
			\hline
			& \multicolumn{6}{c}{\textsc{RMSE}}\\
			\textsc{A$\backslash$T} & L-N & L-G & We & Pa & Ga & E. Pa.\\
			\hline
			L-N & 205.64 & 142.70 & 97.14 & 632.16 & 142.68 & 304.37\\
			L-G & 217.06 & 143.74 & 88.36 & 489.20 & 127.64 & 360.71\\
			We & 194.86 & 149.82 & 76.73 & 233.29 & 103.96 & 292.84\\
			Pa & 209.12 & 147.00 & 79.63 & 493.96 & 105.37 & 305.72\\
			Ga & 204.40 & 146.70 & 76.95 & 251.56 & 103.61 & 310.20\\
			E. Pa. & 583.73 & 172.87 & 76.95 & 1,255.97 & 106.77 & 671.77\\
			4-par. & 802.98 & 187.32 & 76.86 & 2,046.01 & 105.02 & 877.81\\
			5-par. & 445.44 & 159.55 & 77.27 & 719.91 & 104.17 & 994.71\\
			5-par. 2 & 378.09 & 164.44 & 76.85 & 552.24 & 105.09 & 564.65\\
			6-par & 611.29 & 202.13 & 77.06 & 1,254.26 & 105.26 & 502.11\\
			\hline
			\hline
		\end{tabular}
		\caption{Bias and RMSE on the upper $95\%$ reserve from each of the applied (A) distributions and each distribution from Table \ref{tab:speccases} as the true (T) one when $\lambda=1000$ and $n=50$.}
	\label{tab:res.x.0.95.1000.l}
\end{center}
\end{table}

\begin{table}[t]
	\begin{center}
		\begin{tabular}{lcccccc}
			\hline
			\hline
			& \multicolumn{6}{c}{\textsc{Bias}}\\
			\textsc{A$\backslash$T} & L-N & L-G & We & Pa & Ga & E. Pa.\\
			\hline
			L-N & 38.71 & 9.44 & 51.42 & 464.55 & 82.16 & 27.93\\
			L-G & 32.59 & 4.59 & 35.23 & 401.23 & 56.32 & 91.35\\
			We & -67.00 & 9.34 & -0.99 & -74.62 & 3.56 & -7.84\\
			Pa & 19.52 & 19.65 & 21.08 & 197.96 & 16.69 & -58.51\\
			Ga & 96.10 & -2.60 & 0.23 & -60.93 & 1.61 & -76.01\\
			E. Pa. & 532.12 & 23.82 & 0.44 & 640.07 & 6.21 & 303.66\\
			4-par. & 603.33 & 44.13 & -0.53 & 1,307.18 & 5.05 & 482.09\\
			5-par. & 331.77 & 11.94 & -1.06 & 401.44 & 2.86 & 1,632.20\\
			5-par. 2 & 203.55 & 19.97 & -0.54 & 308.52 & 5.01 & 217.32\\
			6-par & 831.41 & 68.95 & -0.29 & 1,279.52 & 5.52 & 250.79\\
			\hline
			& \multicolumn{6}{c}{\textsc{RMSE}}\\
			\textsc{A$\backslash$T} & L-N & L-G & We & Pa & Ga & E. Pa.\\
			\hline
			L-N & 217.17 & 147.23 & 101.21 & 799.52 & 150.22 & 322.07\\
			L-G & 244.53 & 148.72 & 91.57 & 913.80 & 134.35 & 452.60\\
			We & 201.83 & 154.23 & 78.41 & 247.45 & 106.58 & 308.23\\
			Pa & 230.12 & 152.05 & 83.37 & 1,346.81 & 109.04 & 341.41\\
			Ga & 211.16 & 150.92 & 78.65 & 264.53 & 106.25 & 325.22\\
			E. Pa. & 1,523.81 & 203.41 & 78.68 & 6,470.25 & 112.34 & 1,911.60\\
			4-par. & 3,027.09 & 262.63 & 78.54 & 12,924.03 & 108.72 & 2.84\\
			5-par. & 2,842.04 & 173.23 & 78.96 & 2,667.00 & 106.85 & 31,310.46\\
			5-par. 2 & 1,161.21 & 185.25 & 78.52 & 2,268.97 & 111.26 & 2,020.67\\
			6-par & 16,704.91 & 813.31 & 78.72 & 34,387.69 & 120.78 & 1,767.23\\
			\hline
			\hline
		\end{tabular}
		\caption{Bias and RMSE on the upper $99\%$ reserve from each of the applied (A) distributions and each distribution from Table \ref{tab:speccases} as the true (T) one when $\lambda=1000$ and $n=50$.}
	\label{tab:res.x.0.99.1000.l}
\end{center}
\end{table}

\section{Real data example}\label{sec:example}

In this section, we will apply our model as a claim severity
distribution to a set of motor insurance losses from a Norwegian
insurance company. These are discussed in \cite{bolviken2014},
and consist of $6,446$ claims from the years ... to ... The 
deductible is subtracted from the claims, which are given in NOK.
Further, the average claim intensity is $5.7\%$. We have fitted 
each of the ten distributions considered in the simulation study, 
and estimated the $95$ and $99\%$ reserves for each of them, when 
we assume an expected number of $\lambda=1,000$ claims during the 
next year. We have also estimated the $95$ and $99\%$ quantiles of 
the claim severity distribution for interpretation purposes. Each 
quantile estimate is accompanied by a $95\%$ confidence interval 
obtained from $m_{b}=1,000$ bootstrap samples from the given 
distribution. These are compared to the corresponding empirical 
estimates from the data. The results are shown in Table
\ref{tab:data.1.res}. 

We see that the reserve estimates from the five-parameter 2 and the 
extended Pareto distribution are the ones that are the closest to 
the empirical ones. However, the empirical estimate is within the 
confidence interval from almost all the distributions. The five- and 
six-parameter distributions all have confidence intervals that cover 
the empirical estimates, but the five-parameter 2 gives the estimate 
with the smallest uncertainty. The four-parameter distribution, on 
the other hand, over-estimates the two reserves. For the quantiles of
the claim size distribution, the log-normal and the five-parameter 2
distribution give the estimates that are the closest to the empirical 
ones. All the versions of our model have confidence intervals that
cover the empirical estimate of the $95\%$ quantile, but the 
four-parameter distributions significantly over-estimates the $99\%$
quantile, which may explain why its estimates of the reserve are
too large. Further, the quantiles are over-estimated by the log-Gamma
and partly by the Pareto, whereas they are under-estimated by the 
Weibull, Gamma and extended Pareto distributions. To test the quality
of the quantile estimates, we also perform a binomial back-test.
The results from the test are shown in Table \ref{tab:data.2.res}.
These confirm the impression that the five- and six-parameter 
distributions manage to capture the tail behaviour of the data.
All in all, the five- and six-parameter distributions therefore seem 
to be good models for this data set, and particularly the five-parameter 
2, which is consistent with our findings from Section \ref{sec:simstudy}.

\begin{table}[t]
\begin{center}
\begin{tabular}{lcccc}
\hline
\hline
 & \multicolumn{2}{c}{\textsc{Quant. of severity distr.}} & \multicolumn{2}{c}{\textsc{Reserve}}\\
\textsc{Distr.} & $95\%$ & $99\%$ & $95\%$ & $99\%$\\
\hline
\bf Emp & \bf 72.5 & \bf 139.9 & \bf 25.9 & \bf 26.7\\
L-N & 73.9 (71.3, 76.5) & 140.8 (134.7, 146.9) & 26.4 (25.6, 27.2) & 27.3 (26.5, 28.1)\\
L-G & 91.7 (87.7, 96.2) & 231.4 (217.4, 247.2) & 32.8 (31.2, 34.4) & 35.8 (33.8, 38.0)\\
We & 68.5 (66.7, 70.2) & 102.2 (99.1, 105.2) & 25.7 (25.2, 26.3) & 26.5 (25.9, 27.0)\\
Pa & 74.3 (71.9, 76.7) & 122.9 (117.0, 128.9) & 25.7 (25.0, 26.3) & 26.5 (25.8, 27.2)\\
Ga & 65.1 (63.5, 66.6) & 96.3 (93.7, 98.7) & 25.6 (25.0, 26.1) & 26.3 (25.7, 26.8)\\
E. Pa. & 68.5 (65.6, 71.0) & 131.1 (121.8, 139.1) & 25.7 (24.9, 26.5) & 26.6 (25.8, 27.5)\\
4-par. & 69.1 (65.5, 72.6) & 160.5 (144.5, 179.1) & 28.2 (26.6, 29.9) & 31.3 (28.8, 34.5)\\
5-par. & 69.4 (65.9, 86.9) & 146.4 (130.9, 192.7) & 26.5 (24.8, 30.3) & 27.8 (25.7, 32.5)\\
5-par. 2 & 70.3 (67.1, 73.7) & 143.0 (131.4, 154.2) & 26.0 (25.1, 27.0) & 27.0 (26.0, 28.1)\\
6-par. & 69.6 (65.8, 84.6) & 147.9 (131.0, 185.8) & 26.5 (24.7, 29.6) & 27.8 (25.7, 31.7)\\
\hline
\hline
\end{tabular}
\caption{Estimated $95$ and $99\%$ quantiles of the claim size distribution and reserves from each of the ten distributions considered in the simulation study, as well as empirical estimates, along with $95\%$ confidence intervals in the first two columns. The reserves are divided by $1,000$.}
\label{tab:data.1.res}
\end{center}
\end{table}

\begin{table}[t]
\begin{center}
\begin{tabular}{lcccc}
\hline
\hline
 & \multicolumn{2}{c}{\textsc{No. of exc.}} & \multicolumn{2}{c}{\textsc{p-value}}\\
\textsc{Distr.} & $95\%$ & $99\%$ & $95\%$ & $99\%$\\
\hline
L-N & 314 & 64 & 0.668 & 1.000\\
L-G & 212 & 11 & 0.000 & 0.000\\
We &  361 & 162 & 0.030 & 0.000\\
Pa & 310 & 102 & 0.511 & 0.000\\
Ga & 393 & 195 & 0.000 & 0.000\\
E. Pa. & 361 & 85 & 0.030 & 0.012\\
4-par. & 355 & 43 & 0.063 & 0.006\\
5-par. & 351 & 58 & 0.103 & 0.453\\
5-par. 2 & 346 & 61 & 0.179 & 0.754\\
6-par. & 348 & 57 & 0.145 & 0.381\\
\hline
\hline
\end{tabular}
\caption{Numbers of rejections and p-values from the back-test of of the quantile estimates.}
\label{tab:data.2.res}
\end{center}
\end{table}

\section{Conclusion}\label{sec:conclusion}

We propose a new class of claim severity distributions with six 
parameters, and show that this model has the standard two-parameter 
distributions, the log-normal, the log-Gamma, the Weibull, the 
Gamma and the Pareto, as special cases. This distribution is much 
more flexible than its special cases, and therefore more able to 
to capture important characteristics of claim severity data. On 
the other hand, an increased number of parameters usually leads 
to a larger uncertainty in parameter estimates. Therefore, we have 
investigated how this parameter uncertainty affects the 
estimate of the reserve, which is one of the nost important risk
measures within non-life insurance. This is done in a large 
simulation study, where we vary both the characteristics of the
claim size distributions and the sample size. We have also tried
our model on a set of motor insurance claims from a Norwegian
insurance company.

In all we simulate claim size data from ten different distributions, 
including the six special cases and four versions of our model,
each time fitting all the ten distributions to the data. Further,
we estimate the corresponding $95$ and $99\%$ quantiles of the
fitted distributions, and also the $95$ and $99\%$ reserves, where
the total loss distributions follows the collective risk model with
Poisson distributed claim numbers. The quality of the resulting 
estimates is assessed by the bias and the RMSE.

The results from the study show that as long as the amount of data
is reasonable, the five- and six-parameter versions of our model 
provide very good estimates of both the quantiles of the claim 
severity distribution and the reserves, for claim size distributions 
ranging from medium to very heavy tailed. When the sample size is
small, our model appears to struggle with heavy-tailed data, but is
still adequate for data with more moderate tails. The impression 
from the fit to the real data set is the same. Further, the 
performance of the five-parameter 2 distribution, obatined when the
parameter $\tau$ is set to $1$, is overall just as good as the
six-parameter one. This sixth parameter therefore does not seem
to provide a flexibility that is needed. On the other hand, the 
performance of the four-parameter distribution, where the parameter
$\gamma$ has also been set to $1$, is not as consistently good,
not even for smaller sample sizes. Hence, this fifth parameter
really seems to improve the distribution.
 
The poor performance of our model for small data sets following
heavy-tailed distributions is worth a closer look. Of course, 
estimates of quantiles far out in the right tail will inevitably 
be quite uncertain for such data. Still, one might obtain more 
reliable parameter estimates in these cases by adapting the 
estiamtion procedure. One could for instance try to put more
emphasis on the right tail by using a weighted maximum 
likelihood estimator with an adequate weight function. This is
a subject for further work.

\section{Acknowledgements}\label{sec:acknowledgements}
\section*{Appendix 1}\label{sec:appendix}
Let $D_\alpha(z),D_\theta(z),D_\beta(z),D_\tau(z),D_\gamma(z),D_\eta(z)$ 
be the partial derivatives 
of $\log\{f(z)\}$ with respect to $\alpha,\theta,\beta,\tau,\gamma,\eta$
which is straightforward to derive from
in~(\ref{loglike1}) and~(\ref{loglike2}). Adding their values
$D_\alpha(z_i)$, $D_\theta(z_i)$ and so forth over all historical losses
$z_1,\dots,z_n$ yield the gradient vector of the log-likelihood function
that can assist optimization. Introduce
\begin{displaymath}
v = 1+\frac{z}{\beta} \hspace{1cm}\mbox{ and }\hspace*{1cm} 
w = v^{\frac{1}{\gamma}}.
\end{displaymath}
and write $\psi(x)=d\log\{\Gamma(x)\}/dx$ for the so-called Digamma function.
Then, 
\begin{displaymath}
D_\alpha(z)= \psi(\alpha+\theta)-\psi(\alpha)-\frac{\theta}{\alpha}-\log\left(\frac{\theta}{\alpha}\tau^\frac{1}{\eta}(w-1)^\frac{1}{\eta}+1\right)
 +\frac{\alpha+\theta}{\alpha}\cdot\frac{(w-1)^\frac{1}{\eta}}{(w-1)^\frac{1}{\eta}+\alpha/\theta\tau^\frac{1}{\eta}}
\end{displaymath}
\begin{align*}
D_\theta(z)=  \psi(\alpha+\theta)-\psi(\theta)+\log(\theta)+1+\frac{\log(\tau)}{\eta}-\log(\alpha)
 +\frac{\log(w-1)}{\eta}&\\-\log\left(\frac{\theta}{\alpha}\tau^{\frac{1}{\eta}}(w-1)^{\frac{1}{\eta}}+1\right)
 -\frac{\alpha+\theta}{\theta}\cdot\frac{(w-1)^{\frac{1}{\eta}}}{(w-1)^{\frac{1}{\eta}}+\alpha/\theta\tau^{\frac{1}{\eta}}}&
\end{align*}
\begin{displaymath}
D_\beta(z)= 
-\frac{1}{\beta}-\frac{(1-\gamma)z}{\beta^{2}\gamma v}-\left(\frac{\theta}{\eta}-1\right)\frac{zw}{\gamma\beta^{2}v(w-1)}
+\frac{\alpha+\theta}{\beta^{2}\eta\gamma}\cdot\frac{zw(w-1)^{\frac{1}{\eta}-1}}{v\left((w-1)^{\frac{1}{\eta}}+\alpha/\theta\tau^{\frac{1}{\eta}}\right)}\\
\end{displaymath}
\begin{displaymath}
D_\tau(z) =  \frac{\theta}{\eta\tau}-\frac{\alpha+\theta}{\eta\tau}\cdot\frac{(w-1)^{\frac{1}{\eta}}}{(w-1)^{\frac{1}{\eta}}+\alpha/\theta\tau^{\frac{1}{\eta}}}\\
\end{displaymath}
\begin{displaymath}
D_\gamma(z) 
=  -\frac{1}{\gamma}-\frac{\log(v)}{\gamma^{2}}-\left(\frac{\theta}{\eta}-1\right)\frac{w}{\gamma^{2}(w-1)}\log(v)\\
 +\frac{\alpha+\theta}{\gamma^{2}\eta}\frac{(w-1)^{\frac{1}{\eta}-1}}{(w-1)^{\frac{1}{\eta}}+\alpha/\theta\tau^{\frac{1}{\eta}}}w\log(v),
\end{displaymath}
\begin{displaymath}
D_\eta(z) =  -\frac{\theta\log(\tau)}{\eta^{2}}-\frac{1}{\eta}-\frac{\theta}{\eta^{2}}\log(w-1)
+\frac{\alpha+\theta}{\eta^{2}}(\log(w-1)+\log(\tau))\frac{(w-1)^{\frac{1}{\eta}}}{(w-1)^{\frac{1}{\eta}}+\alpha/\theta\tau^{\frac{1}{\eta}}}\\
\end{displaymath}

\bibliography{references}
\bibliographystyle{elsarticle-harv}

\end{document}